\documentclass[11pt]{article}
\usepackage[margin=1in]{geometry}
\usepackage{subcaption}
\usepackage{color}
\usepackage{hyperref}
\usepackage[dvipsnames,svgnames,x11names,hyperref, enumerate]{xcolor}
\definecolor{blueish}{HTML}{007F99}
\definecolor{purpleish}{HTML}{72177A}
\hypersetup{
	colorlinks,
	citecolor=purpleish,
	filecolor=red,
	linkcolor=blueish,
	urlcolor=Blue
}
\usepackage{graphicx}
\usepackage{amsmath, amssymb, amsthm} 
\usepackage{url}            
\usepackage{booktabs}       
\usepackage{amsfonts}       
\usepackage{thm-restate}

\newtheorem{theorem}{Theorem}
\newtheorem{prop}[theorem]{Proposition}

\newtheorem{lemma}[theorem]{Lemma}
\newtheorem{cor}[theorem]{Corollary}

\newtheorem{defn}{Definition}

\renewcommand{\epsilon}{\varepsilon}

\renewcommand{\le}{\leqslant}
\renewcommand{\ge}{\geqslant}
\renewcommand{\leq}{\leqslant}
\renewcommand{\geq}{\geqslant}
\renewcommand{\P}{\mathbf{P}}

\newcommand{\de}{d_\mathrm{E}}

\newcommand{\remove}[1]{}

\newcommand{\HLinkShort}[2]{\hyperref[#2]{#1\ref*{#2}}}
\newcommand{\HLink}[2]{\hyperref[#2]{#1~\ref*{#2}}}
\newcommand{\HLinkPage}[2]{\hyperref[#2]{#1~\ref*{#2}%
		$_\text{p\pageref{#2}}$}}
\newcommand{\HLinkPageOnly}[1]{\hyperref[#1]{Page~\refpage*{#1}%
		$_\text{p\pageref{#1}}$}}

\newcommand{\HLinkSuffix}[3]{\hyperref[#2]{#1\ref*{#2}{#3}}}
\newcommand{\HLinkPageSuffix}[3]{\hyperref[#2]{#1\ref*{#2}%
		#3$_\text{p\pageref{#2}}$}}

\newcommand{\seclab}[1]{\label{section:#1}}
\newcommand{\secref}[1]{\HLink{Section}{section:#1}}

\newcommand{\corlab}[1]{\label{cor:#1}}
\newcommand{\corref}[1]{\HLink{Corollary}{cor:#1}}%

\providecommand{\deflab}[1]{\label{def:#1}}
\newcommand{\defref}[1]{\HLink{Definition}{def:#1}}

\newcommand{\apndlab}[1]{\label{apnd:#1}}
\newcommand{\apndref}[1]{\HLink{Appendix}{apnd:#1}}

\newcommand{\lemlab}[1]{\label{lemma:#1}}
\newcommand{\lemref}[1]{\HLink{Lemma}{lemma:#1}}%

\newcommand{\tablab}[1]{\label{table:#1}}%
\newcommand{\tabref}[1]{\HLink{Table}{table:#1}}%

\newcommand{\proplab}[1]{\label{propo:#1}}%
\newcommand{\propref}[1]{\HLink{Proposition}{propo:#1}}%

\newcommand{\propertylab}[1]{\label{property:#1}}%
\newcommand{\propertyref}[1]{\HLink{Property}{property:#1}}%

\newcommand{\thmlab}[1]{{\label{theo:#1}}}
\newcommand{\thmref}[1]{\HLink{Theorem}{theo:#1}}

\providecommand{\eqlab}[1]{}%
\renewcommand{\eqlab}[1]{\label{equation:#1}}
\newcommand{\Eqlab}[1]  {\label{equation:#1}}
\newcommand{\Eqref}[1]{\HLinkSuffix{Eq.~(}{equation:#1}{)}}

\usepackage{dsfont}
\usepackage{booktabs, xcolor, colortbl}
\begin{document}
	\title{Approximate Trace Reconstruction}
	\author{
	}
	\date{\today}
	
	\author{Sami Davies%
 		\thanks{%
 			{University of Washington (\url{daviess@uw.edu}).}}
 		\and Mikl\'os Z. R\'acz%
 		\thanks{%
 			{Princeton University (\url{mracz@princeton.edu}); research supported in part by NSF grant DMS 1811724 and by a Princeton SEAS Innovation Award.}}
 		\and Cyrus Rashtchian%
 		\thanks{%
 			{Dept. of Computer Science \& Engineering, University of California, San Diego (\url{crashtchian@eng.ucsd.edu})}.}
     		\and Benjamin G. Schiffer%
 		\thanks{%
 			{Princeton University (\url{bgs3@princeton.edu}).}
             } }
	
	\maketitle

\begin{abstract}%
In the usual trace reconstruction problem, the goal is to exactly reconstruct an unknown string of length $n$ after it passes through a deletion channel many times independently, producing a set of traces (i.e., random subsequences of the string). We consider the relaxed problem of approximate reconstruction. Here, the goal is to output a string that is close to the original one in {\em edit distance} while using much fewer traces than is needed for exact reconstruction. We present several algorithms that can approximately reconstruct strings that belong to certain classes, where the estimate is within $n/\mathrm{polylog}(n)$ edit distance, and where we only use $\mathrm{polylog}(n)$ traces (or sometimes just a single trace). These classes contain strings that require a linear number of traces for exact reconstruction and which are quite different from a typical random string. From a technical point of view, our algorithms approximately reconstruct consecutive substrings of the unknown string by aligning dense regions of traces and using a run of a suitable length to approximate each region. To complement our algorithms, we present a general black-box lower bound for approximate reconstruction, building on a lower bound for distinguishing between two candidate input strings in the worst case. In particular, this shows that approximating to within $n^{1/3 - \delta}$ edit distance requires $n^{1 + 3\delta/2}/\mathrm{polylog}(n)$ traces for $0< \delta < 1/3$ in the worst case. 
\end{abstract}

\allowdisplaybreaks

\section{Introduction}

In the \emph{trace reconstruction} problem, we observe noisy samples of an unknown binary string after passing it through a deletion channel several times~\cite{BatuKannan04-RandomCase,Levenshtein:2001}. For a parameter $q\in (0,1)$, the channel deletes each bit of the string with probability $q$ independently, resulting in a {\em trace}. The deletions for different traces are also independent. We only observe the concatenation of the surviving bits, without any information about the deleted bits or their locations. 

How many samples (traces) from the deletion channel does it take to exactly recover the unknown string with high probability? Despite two decades of work, this question is still wide open. 
For a worst-case string, very recent work shows that $\exp(\widetilde{O}(n^{1/5}))$ traces suffice~\cite{chase2020new}, building upon the previous best bound of 
$\exp(O(n^{1/3}))$~\cite{de2019optimal, NazarovPeres17-WorseCase}; 
furthermore, $\widetilde \Omega(n^{3/2})$ traces are necessary~\cite{Chase19, HL18}. 
Improved upper bounds are known in the {\em average-case} setting, where the unknown string is uniformly random~\cite{BatuKannan04-RandomCase,HolensteinMPW08,HoldenPemantlePeres18-RandomCase,tr-revisited, PeresZ17, viswanathan}, in the {\em coded} setting, where the string is guaranteed to reside in a pre-defined set~\cite{BLS19, CGMR, sabary2020error, srinivasavaradhan2018maximum, srinivasavaradhan2020algorithms}, and in the smoothed-analysis setting where the unknown string is perturbed before trace generation~\cite{chen2020polynomialtime}. 

Given that exact reconstruction may be challenging, we relax the requirements and ask: when is it possible to {\em approximately} reconstruct an unknown string with much less information? More precisely, the algorithm should output a string that is close to the true string under some metric. Since the channel involves deletions, we consider {\em edit distance}, measuring the minimum number of insertions, deletions, and substitutions between a pair of strings. Letting $n$ denote the length of the unknown string, we investigate the necessary and sufficient number of traces to approximate the string up to $\varepsilon n$ edit distance. We call this the {\em $\varepsilon n$-approximate reconstruction} problem.

Trace reconstruction has received much recent attention because of DNA data storage, where reconstruction algorithms are used to recover the data~\cite{OAC,church2012next,bhardwaj2020,goldman,yazdi2017portable, lopez2019dna}. Biochemical advances have made it possible to store digital data using synthetic DNA molecules with higher information density than electromagnetic devices. 
During the data retrieval process, each DNA molecule is imperfectly replicated several times, leading to a set of noisy strings that contain insertion, substitution, and deletion errors.
Error-correcting codes are used to deal with missing data, and hence, an approximate reconstruction algorithm would be practically useful. Decreasing the number of traces would reduce the time and cost of data retrieval. 

For any deletion probability $q$, a single trace achieves a $qn$-approximation in expectation. On the other hand, if $q = 1/2$, then it is not clear whether $(n/100)$-approximate reconstruction requires asymptotically fewer traces than exact reconstruction.
More generally, we propose the following goal: determine the smallest $\varepsilon$ such that any string can be $\varepsilon n$-approximately reconstructed with $\mathrm{poly}(n)$ traces. Here $\varepsilon$ is a parameter that may depend on $n$ and $q$. Although we focus on an information-theoretic formulation (measuring the number of traces), the reconstruction algorithm should also be computationally efficient (polynomial time in $n$ and the number of traces). 

A natural approach would be to transform exact reconstruction methods into more efficient approximation algorithms. Unfortunately, revising these algorithms to allow some incorrect bits seems nontrivial or perhaps impossible. For example, certain algorithms assume that the string has been perfectly recovered up to some point, and they use this to align the traces and determine the next bit~\cite{BatuKannan04-RandomCase,HolensteinMPW08,HoldenPemantlePeres18-RandomCase}. Another technique involves computing the bit-wise average of the traces and outputing the string that most closely matches the average. These {\em mean-based} statistics suffice to distinguish any pair of strings, but only if there are $\exp(\Omega(n^{1/3}))$ traces~\cite{de2019optimal, NazarovPeres17-WorseCase}. Also, the maximum likelihood solution is poorly understood for the deletion channel, and current analyses are limited to a small number of traces~\cite{mitzenmacher-survey,sabary2020error,srinivasavaradhan2018maximum, srinivasavaradhan2020algorithms}

Designing algorithms to find approximate solutions may in fact require fundamentally different methods than exact reconstruction. For a simple supporting example, consider the family of strings containing all ones except for a single zero that lies in some position between $n/3$ and $2n/3$, e.g.,  
$111\cdots 11 0 11\cdots 111$.
Determining the exact position of the zero requires $\Omega(n)$ traces when the deletion probability is a constant~\cite{BatuKannan04-RandomCase,tr-revisited}. However if the string comes from this family, we can output the all ones vector and achieve an approximation to within Hamming distance one. 

As a starting point, we consider classes of strings defined by run-length assumptions. For instance, if the 1-runs are sufficiently long and the zero runs are either short or long, we can $\epsilon n$-approximately reconstruct the string using $O(\log(n)/\epsilon^2)$ traces.  
We then strengthen our upper bound to work even when the string can be partitioned into regions that are either locally dense or sparse. 
Finally, we prove new lower bounds on the necessary trace complexity; for example, approximating arbitrary strings to within $n^{1/3 - \delta}$ edit distance requires $n^{1 + 3\delta/2}/\mathrm{polylog}(n)$ traces for any constant $0< \delta < 1/3$.


\subsection{Related work}\seclab{Related}

The trace reconstruction problem was introduced to the theoretical computer science community by Batu, Kannan, Khanna, and McGregor~\cite{BatuKannan04-RandomCase}. 
There is an exponential gap between the known upper and lower bounds for the number of traces needed to reconstruct an arbitrary string with constant deletion rate~\cite{Chase19,de2019optimal,chase2020new,HL18, NazarovPeres17-WorseCase}. The main open theoretical question is whether a polynomial number of traces suffice. There has also been experimental and theoretical work on maximum likelihood decoding, where approximation algorithms have been developed for average-case strings given a constant number of traces~\cite{sabary2020error,srinivasavaradhan2018maximum,srinivasavaradhan2020algorithms}.~Holden, Pemantle, and Peres show that $\exp(O (\log^{1/3} n))$ traces suffice for reconstructing a random string, building on previous work~\cite{BatuKannan04-RandomCase,HolensteinMPW08, HoldenPemantlePeres18-RandomCase, PeresZ17, viswanathan}.
Chase extended work by Holden and Lyons to show that $\widetilde{\Omega}(\log^{5/2} n)$ traces are necessary for random strings~\cite{Chase19, HL18}. 

A related question to ours is to understand the limitations of known techniques for distinguishing pairs of strings that are close in edit distance. Grigorescu, Sudan, and Zhu show that there exist pairs that cannot be distinguished with a small number of traces when using a mean-based algorithm~\cite{grigorescu2020limitations}. They further identify strings that are separated in edit distance, yet can be exactly reconstructed with few traces. Their results are incomparable to ours because the sets of strings they consider are different. Instead of considering local density assumptions, they consider local agreement up to single-bit shifts. Their algorithm uses a subexponential number of traces only when the edit distance separation is at most $o(\sqrt{n})$.

Many other variants of trace reconstruction have been studied as stand-alone problems, united by the goal of broadening our understanding of reconstruction problems. 
Krishnamurthy, Mazumdar, McGregor, and Pal consider matrices (rows/columns deleted) and sparse strings~\cite{KMMP19}. 
Davies, R\'acz, and Rashtchian consider labeled trees,
where the additional structure of some trees leads to more efficient reconstruction~\cite{DRR19}. Circular trace reconstruction considers strings and traces up to circular rotations of the bits~\cite{narayanan2020circular}.
Population recovery reconstructs multiple unknown strings simultaneously~\cite{Ban, BCSS19, narayanan2020population}. Going beyond i.i.d.~deletions, algorithms have also been developed for position- or character- dependent error rates~\cite{HartungHP18}, or for ancestral state reconstruction, where deletions are based on a Markov chain~\cite{andoni2012global}. 
It should not go without mention that forms of approximate trace reconstruction have been studied in more applied frameworks; in particular Srinivasavaradhan, Du, Diggavi, and Fragouli study heuristics for reconstructing approximately given one or two traces~\cite{srinivasavaradhan2018maximum}.

\paragraph{Comparison to Coded Trace Reconstruction.} 
Cheraghchi, Gabrys, Milenkovic, and Ribeiro explore coded trace reconstruction, where the unknown string is assumed to come from a code, and they show that codewords can be reconstructed with high probability using much fewer traces than average-case reconstruction~\cite{CGMR} (see also~\cite{drinea2007improved, Levenshtein:2001, mitzenmacher-survey}).
Brakensiek, Li, and Spang extend this work and present codes with rate $1- \gamma$ that can be reconstructed
using $\exp(O(\log^{1/3}(1/\gamma)))$ traces~\cite{BLS19}.
Improved coded reconstruction results are known when the number of errors in a trace is a fixed constant~\cite{abroshan2019coding, chrisnata2020optimal, haeupler2014repeated, kiah2020coding, sabary2020error}.

An existing approach for coded trace reconstruction does use approximation as an intermediate step, where the original string can be recovered after error correction~\cite{BLS19}. Our focus is  different, and our results are incomparable to those from coded trace reconstruction. We investigate classes of strings that are very different from codes (e.g., pairs of strings in our classes can be very close). We also consider strings that require $\Omega(n)$ traces to exactly reconstruct, whereas the work on coded trace reconstruction shows that their classes of strings can be exactly reconstructed with a sublinear number of traces. Overall, we do not aim to optimize the ``rate'' of our classes of strings. Instead, our main contribution is the effectiveness of new algorithmic techniques and local approximation methods, including novel alignment ideas and the use of runs in approximating edit distance.

Additionally, coded trace reconstruction lower bounds can be used as a black box to obtain lower bounds for approximate trace reconstruction by constructing a code that is an $\varepsilon n$-net~\cite{BLS19}. However, these lower bounds reduce to results on average-case reconstruction, and hence, this approach currently leads to lower bounds for approximate reconstruction that are exponentially smaller than what we prove.

\subsection{Our results}

We assume that the deletion probability $q$ is a fixed constant and $p := 1-q$ is the retention probability.   In our statements, $C,C',C'',C_1,C_2,\ldots$ denote constants, and $O(\cdot)$ hides constants, where these may depend on $p, q$. 
Unless stated otherwise, $\log(\cdot)$ has base $1/q$. The phrase \emph{with high probability} means probability at least $1-O(1/n)$.  A {\em run} in a string is a substring of consecutive bits of the same value, and we often refer specifically to $0$-runs and $1$-runs. We use bold $\mathbf{r}$ to denote runs, or more generally substrings, and let $|\mathbf{r}|$ denote its length (number of bits). \tabref{results} summarizes our results, and we restate the theorems in the relevant sections for the reader's convenience.

\subsubsection*{Algorithms for approximate reconstruction}

Our results exhibit the ability to approximately reconstruct strings based on various run-length or density assumptions. For these classes of strings, we develop new polynomial-time, alignment-based algorithms, and we show that $O(\log (n) / \epsilon^2)$ traces suffice. 
We assume that the algorithms know $n$, $q$, $\epsilon$, and the class that the unknown string comes from, though the last assumption is not necessary for \thmref{gapone}. 
We also provide warm-up examples (see \propref{longruns} and \propref{oneruns} in \secref{warm-up}), which may be helpful to the reader before diving into the algorithms in \secref{upper-bounds}. 


Our first theorem only requires $1$-runs to be long, while the length of the $0$-runs is more flexible; they can be either long or short, assuming there is a gap. 

\begin{restatable}[]{theorem}{gapone}
    \thmlab{gapone}
    Let $X$ be a string on $n$ bits such that all of its $1$-runs have length at least $C' \log(n) / \varepsilon$ and 
    none of its $0$-runs have length between 
        $C' \log(n)$ and 
    $3C'\log(n)$. 
    	There exists some constant $C$ such that if $C' \geq C$, then 
     $X$ can be $\varepsilon n$-approximately reconstructed with $O(\log(n) / \varepsilon^2)$ traces. 
\end{restatable}

\begin{table}[t] 	\renewcommand{\arraystretch}{1.2} \footnotesize
	\centering 
	\caption{Table of sample complexity bounds for $\varepsilon n$-approximate reconstruction.}
	\begin{tabular}{l c l} 
		\toprule 
		Classes of strings & \# samples $\varepsilon n$-approx. & Reference \\
		\midrule
		All runs have length $\geq 5 \log(n)$ & $O( \log(n) / \varepsilon^2)$ & \propref{longruns} \& \corref{longrunsrobust}  \\ 
		The $1$-runs have length $\geq C' \log(n) / \varepsilon^2 $ & 1  & \propref{oneruns} \\
		Long $1$-runs; either long or short $0$-runs & $O( \log(n) / \varepsilon^2)$ & \thmref{gapone} \& \thmref{gaponerobust} \\
		Intervals length $\geq C' \log(n)/\varepsilon^2$, density $\geq 1 - \frac{\varepsilon}{12}$  & 1 & \thmref{singletracetwo} \\
			\midrule
			Arbitrary strings, $n^{1/3 - \delta}$  edit distance, $\delta \in (0,\frac{1}{3})$  &
			   $\widetilde \Omega( n^{1 + 3\delta/2})$ & \thmref{blackboxlowerbound} \& \corref{LBCorollary} \\
		\bottomrule
	\end{tabular}
	\label{table: results}\tablab{results}
\end{table}
\noindent

The following theorem extends \thmref{gapone} to a wider class of strings by allowing many of the bits in the runs to be arbitrarily flipped to the opposite value.

\begin{restatable}[]{theorem}{gaponerobust}
    \thmlab{gaponerobust} 
    Suppose that $p > 3\epsilon$. 
    Let $Y$ be a string on $n$ bits such that all of its $1$-runs have length at least $ C' \log(n) / \varepsilon$ and none of its $0$-runs have length between $C' \log(n)$ and $3C' \log(n)$. 
    Suppose that $X$ is formed from~$Y$ by modifying at most $\epsilon C' \log(n)$ arbitrary bits in each run of $Y$. 
    If $C' \geq 1000/p$, then $X$ can be $\epsilon n$-approximately reconstructed with $O(\log(n) / \varepsilon^2)$ traces.
\end{restatable}

For the final class, we consider a slightly different relaxation of having long runs. We impose a local density or sparsity constraint on contiguous intervals. If this holds, then a single trace suffices.

\begin{restatable}{theorem}{singletracetwo}
	\thmlab{singletracetwo}
	There exists some constant $C$ such that for $C' \geq C$, 
	if $X$ can be divided into contiguous intervals $I_1,\ldots,I_m$ with all $I_i$ having length at least
	$C' \log(n) / \varepsilon^2$ and density at least $1-\frac{\epsilon}{12}$ of  $0$s or $1$s,
	then $X$ can be $\epsilon n$-approximately reconstructed with a single trace in polynomial time.
\end{restatable}

The algorithm for \thmref{singletracetwo} extends to handle independent insertions at a rate of $O(\epsilon)$, since the proof relies on finding high density regions, which are unchanged by such insertions. 

We provide some justification for the strings considered in the above theorems. Strings that either contain long runs or that are locally dense are a natural class to examine in order to understand the advantage gained by approximate reconstruction over exact. Strings with sufficiently long runs require $\Omega(n)$ traces to reconstruct exactly, as exact reconstruction for this set involves distinguishing between our prior example strings $1^{n/2} 01^{n/2-1}$ and $1^{n/2-1} 0 1^{n/2}$, but can be approximately reconstructed with substantially less traces for large enough values of $\epsilon$. We then relax the condition that strings have long runs to the condition that strings are locally dense.
Both strings with long runs and strings that are locally dense also look very different than average-case strings (i.e., uniformly random), which have runs with length at most $2 \log n$ with high probability and can be exactly reconstructed with $O(\exp(\log^{1/3}(n)))$ traces~\cite{HoldenPemantlePeres18-RandomCase}.


\subsubsection*{Lower bounds for approximate reconstruction}
We prove lower bounds on the number of traces required for $\varepsilon n$-approximate reconstruction. We present two results, for edit distance and Hamming distance, respectively. 
The more challenging result is \thmref{blackboxlowerbound}, which shows that any algorithm that reconstructs a length $n$ arbitrary string within $\epsilon n$ edit distance requires $f(C/\epsilon)$ traces, where $f(n)$ denotes the minimum number of traces for distinguishing a pair of $n$-bit strings. Currently, $f(n) = \widetilde{\Omega}(n^{1.5})$ is the best known lower bound for exact reconstruction, which argues via a pair of strings that are hard to distinguish~\cite{Chase19}.

\begin{restatable}{theorem}{blackboxlowerbound}
\thmlab{blackboxlowerbound}
	Suppose that $f(n)$ traces are required to distinguish between two length $n$ strings $X'$ and $Y'$ with probability at least $1/2 + \alpha$, where $\alpha = 1/8$. Then there exists absolute constants $C, \epsilon^\star> 0$ such that for $\epsilon^\star \geq \epsilon \geq \log(n)/n$, any algorithm that $\epsilon n$-approximately reconstructs  arbitrary length $n$ strings with probability $1 - 1/n$ must use at least $f(C/\epsilon)$ traces.
\end{restatable}
\noindent
Plugging in the bound on $f(1/\epsilon)$, our theorem shows that $(1/\epsilon)^{3/2}/\mathrm{polylog}(1/\epsilon)$ traces are required for $\epsilon n$-approximate reconstruction. For example, if $\epsilon = n^{-2/3 - \delta}$, then we obtain the following. 

\begin{cor}\corlab{LBCorollary} For any constant $\delta \in (0,1/3)$, we have that
	$n^{1 + 3\delta/2}/\mathrm{polylog}(n)$ 
	traces are necessary to $n^{1/3 - \delta}$-approximately reconstruct an arbitrary $n$-bit string with probability $1-1/n$.
\end{cor} 
\noindent
\thmref{blackboxlowerbound} also allows for $\epsilon$ to be as small as $\log(n)/n$, implying that a very close approximation is not possible with substantially fewer traces than exact reconstruction.  

Our lower bound for Hamming distance in \thmref{lowerboundHamming} is simpler. It shows that $\Omega(n)$ traces are necessary to achieve an approximation closer than $n/4$ in Hamming distance to the actual string. In particular, we get a linear lower bound for a linear Hamming distance approximation, which is much stronger than our result for edit distance.

\begin{restatable}{theorem}{lowerboundHamming}
	\thmlab{lowerboundHamming}
	Any algorithm that can output an approximation within Hamming distance $n/4 - 1$ of an arbitrary length $n$ string with probability at least $3/4$ must use $\Omega(n)$ traces.
\end{restatable}

\subsection{Technical overview}

The high-level strategy for all of our algorithms is the following. First, we identify the remnants of structured substrings, that is, long runs and dense substrings, from the original string in the traces. Then, when given more than one trace, we can use these substrings to align traces. After aligning traces, we capitalize on the approximate nature of our objective by estimating lengths of runs which are close in edit distance to substrings of the original string.

The gap condition for $0$-runs in \thmref{gapone} states that the unknown string only contains $0$-runs with length either less than $a_1 := C' \log (n) $ or greater than $a_2 := 3C' \log (n) $,
for large enough $C'$ (and nothing in the middle). 
We show that there exist values $a_1', a_2'$, with $p a_1 < a_1' < a_2' < p  a_2$, such that with high probability 
there does not exist a $0$-run of length at least $a_2$ in the original string that has been reduced to a $0$-run of length less than $a_2'$ in a trace, nor a $0$-run of length less than $a_1$ reduced to a $0$-run of length more than $a_1'$.
This implies that we can always distinguish between short and long runs of $0$s in all of the traces (which would be challenging without the gap condition).
We can align long runs of $0$s from the traces and take a scaled average of the lengths of the $i$th long run of $0$s across all $T$ traces.
By using a scaled average across traces, we can estimate the number of bits between consecutive long runs of $0$s. Then, our algorithm outputs 
a run of $1$s here, which accounts for long runs of $1$s and short runs of $0$s. Note that this piece of the algorithm is inherently approximate since we replace  
short runs of $0$s with $1$s. This completes, what we call, our {\em algorithm for identifying long runs}.

The algorithm for \thmref{gaponerobust} is similar to \thmref{gapone}. We identify long $0$-runs from $Y$ in each of the traces and align by these $0$-runs, then approximate the rest using $1$s. However, the alignment step is more difficult since  
the long $0$-runs from $Y$ may not be $0$-runs in $X$ and not easily found in traces. Instead, we identify long 0-dense substrings in each trace that with high probability originate from long $0$-runs in $Y$. We refer to this as the {\em algorithm for identifying dense substrings}. Then we align and average as in \thmref{gapone} to approximate the unknown string. 

Our algorithm in \thmref{singletracetwo} takes a uniform partition of a single trace, where each part has length $C\log(n) / \varepsilon$, and it outputs a series of runs, where each run has length $C\log(n) / (\varepsilon p)$ and parity the majority bit of the interval.
Note the partitions have length at most an $O(\epsilon)$ fraction of the high density intervals.
Therefore in any high density interval of the original string,
most of the partitions of the trace originating from that interval will also have high density of the same bit. We refer to the method for this result as the {\em algorithm for majority voting in substrings}.

The algorithms and analyses for these three theorems are in \secref{upper-bounds}.

\paragraph{Lower Bounds.} For the edit distance approximation in \thmref{blackboxlowerbound}, we start with two strings of length $C/\epsilon$ that require $f(C/\epsilon)$ traces to distinguish for some constant $C \in (0,1)$ and $\epsilon < C$. We then construct a hard distribution over length $n$ strings by concatenating $\epsilon n / C$ substrings, where each substring is an independent random choice between the two strings. Our strategy is to show that if the algorithm outputs an approximation within $\epsilon n$ edit distance, then it must correctly determine a large number of the component strings. However, proving this requires some work because the guarantee of the reconstruction algorithm is in terms of an edit distance approximation. To handle this challenge, we provide a technical lemma that relates the edit distance of any pair of strings to a sum of binary indicator vectors for the equivalence of certain substrings (\lemref{block-ed-lb}). Then, we use this lemma to argue that the algorithm's output must be far from the true string if the number of traces is less than $f(C/\epsilon)$ because many substrings must disagree. 

For the Hamming distance lower bound in \thmref{lowerboundHamming}, we use a more straightforward argument. We start with a known lower bound from Batu, Kannan, Khanna, and McGregor~\cite{BatuKannan04-RandomCase}. They observe that $\Omega(k)$ traces are necessary to determine if a string starts with $k$ or $k+1$ zeros. We then construct a hard pair of strings of length roughly $4k$ such that if the algorithm misjudges the prefix length, then it must incur a cost of at least $2k$ in Hamming distance. Since $k = \Omega(n)$, we obtain the desired lower bound. 

The proofs for both lower bounds appear in \secref{lower-bounds}.

\subsection{Preliminaries}

Let $\de(X,X')$ denote the edit distance between $X$ and $X'$, defined as the minimum number of insertions, deletions, and substitutions that are required to transform $X$ into $X'$. Note that  edit distance is a metric. 
For each class of strings that we consider, we present an algorithm and argue that it can $\varepsilon n$-approximately reconstruct any string from the class.
Our algorithms output a string~$\widehat{X}$, an approximation of $X$, satisfying $\de(X,\widehat X) \leq \epsilon n$ with high probability. 

We denote a single run by $\mathbf{r}$ and a set of runs by ${\mathbf{r}}_1, \ldots, {\mathbf{r}}_k$. 
Our convention is to let $X$ denote the unknown string that we wish to reconstruct, and $Y$ will often be a modified version.
A single trace will be denoted by $\widetilde{X}$ and a set of traces by $\widetilde{X}_1,\ldots, \widetilde{X}_T$.
Tildes will also be used to mark runs and intervals of traces. 
Some strings $X$ we partition into $\ell$ substrings $X^1,\ldots, X^{\ell}$; their concatenation to form $X$ is denoted as $X = X^1 X^2 \cdots X^{\ell}$. 

Some of our algorithms reconstruct $X$ by partitioning it into substrings $X^1,\ldots,X^{\ell}$ and reconstructing these substrings approximately. Specifically, we will find strings $\widehat{X}^i$ such that the edit distance between $\widehat{X}^i$ and $X^i$ is at most $\varepsilon |X^i|$, and then we will invoke the following lemma to see that $X = X^1\cdots X^{\ell}$ and $\widehat{X} = \widehat{X}^1\cdots \widehat{X}^{\ell}$ have edit distance at most $\varepsilon n$. 
\renewcommand{\de}{d_\mathsf{E}}
\begin{lemma}\lemlab{substringrecon}
    Let $X = X^1 X^2 \cdots X^{\ell}$ and $\widehat{X} = \widehat{X}^1\cdots \widehat{X}^{\ell}$ be strings on $n$ bits. If the edit distance between  $X^i$  and $\widehat{X}^i$ is at most $\varepsilon |X^i|$ for all $i \in [\ell]$, then 
    $\de(X,\widehat{X}) \leq \varepsilon n$. 
\end{lemma}
 \begin{proof}
 We will use the fact that edit distance satisfies the triangle inequality. Consider bit strings $ X=X^1X^2$ and $\widehat{X} = \widehat{X}^1\widehat{X}^2$. Then,
 $$
 \de(X^1X^2, \widehat{X}^1\widehat{X}^2) \leq  
 \de(X^1 X^2, \widehat{X}^1 X^2) + \de(\widehat{X}^1 X^2, \widehat{X}^1\widehat{X}^2)   
=   \de(X^1, \widehat{X}^1) + \de(X^2,\widehat{X}^2).
 $$
This extends to $ X=X^1\cdots X^{\ell}$ and $\widehat{X} = \widehat{X}^1\cdots \widehat{X}^{\ell}$ by recursively applying the above inequality. 
\end{proof}

\section{Warm-up: Approximating strings that only have long runs}
\seclab{warm-up}

We begin with two simple cases that demonstrate some of our algorithmic techniques.
For this section, we defer proofs to~\apndref{A}. We note that other methods may lead to similar or slightly better results in some regimes, but we follow this presentation as a prelude to~\secref{general-ub}.

The first algorithm $\epsilon n$-approximately reconstructs a string with long runs using $\Omega(\log(n) / \varepsilon^2)$ traces by scaling an average of the run length across all traces.

\begin{prop}\proplab{longruns}
	Let $X$ be a string on $n$ bits such that all of its runs have length at least $\log(n^5)$.
	Then $X$ can be $\varepsilon n$-approximately reconstructed with 
	$O(\log(n) / \varepsilon^2)$ traces.
\end{prop}

\subsubsection*{Algorithm}

\begin{enumerate}
    \item[{\bf Set-up:}] String $X$ on $n$ bits such that all of its runs have length at least $\log(n^5)$.
    \item Sample $T = \frac{2}{p \varepsilon^2} \log(n)$ traces, $\widetilde{X}_1,\ldots, \widetilde{X}_T$, from the deletion channel with deletion probability~$q$. Fail if all traces do not have the same number of runs. Otherwise let $k$ denote the number runs in every trace.
    \item Compute $\widetilde{\mu}_i = \frac{1}{T}\sum_{j=1}^T |\widetilde{\mathbf{r}}^j_i|$ for all $i \in [k]$, where $\widetilde{\mathbf{r}}^j_1,\widetilde{\mathbf{r}}^j_2,\ldots,\widetilde{\mathbf{r}}^j_k$ are the $k$ runs of $\widetilde{X}_j$.
    \item Output $\widehat{X} = \widehat{X}_1 \cdots \widehat{X}_k$, where $\widehat{X}_i$ has length $ \widetilde{\mu}_i/p$ and bit value matching run $i$ of the traces.
\end{enumerate}

\noindent
The analysis is a basic use of Chernoff bounds; see \apndref{A} for details.

Ideally we would only require that $1$-runs have length $\Omega(\log(n))$, without restricting the length of $0$-runs. The following result shows that if we require the $1$-runs to be $\Omega(\frac{1}{\epsilon^2}\log(n))$, 
which is an order of $1 / \varepsilon$ larger than in \thmref{gapone}, 
then approximate reconstruction is possible using one trace. 

\begin{prop}
    \proplab{oneruns}
    Let $X$ be a string on $n$ bits such that all of its $1$-runs have length at least $C' \log(n) / \varepsilon^2$.
    Then there exists a constant $C$ such that for $C' \geq C$, $X$ can be $\epsilon n$-approximately reconstructed with a single trace.
\end{prop}

\subsubsection*{Algorithm}

\begin{enumerate}
    \item[{\bf Set-up:}] String $X$ on $n$ bits such that all its $1$-runs have length at least $\frac{6}{p} \log(n) / \varepsilon^2$.
    \item Sample 1 trace $\widetilde{X}$ from the deletion channel with deletion probability $q$.
    \item Let $L := \frac{\log(n)}{10 \varepsilon}$; $\widetilde{\mathbf{r}}_1$,\ldots,$\widetilde{\mathbf{r}}_k$ be $0$-runs in $\widetilde{X}$ with length at least $L$; 
    and $\widetilde{\mathbf{s}}_i$, for $i \in \{0,1,\ldots,k+1\}$, be the bits in $\widehat{X}$ before $\widetilde{\mathbf{r}}_1$, between $\widetilde{\mathbf{r}}_{i}$ and $\widetilde{\mathbf{r}}_{i+1}$, and after $\widetilde{\mathbf{r}}_k$, respectively.
    \item Output $\widehat{X} = \widehat{1}_0\widehat{0}_1 \widehat{1}_1\cdots \widehat{1}_k\widehat{0}_k \widehat{1}_{k+1}$, where $\widehat{1}_i$ is a $1$-run, length $\frac{|\widetilde{\mathbf{s}}_i|}{p}$, and $\widehat{0}_i$ is a $0$-run, length~$\frac{|\widetilde{\mathbf{r}}_i|}{p}$.
\end{enumerate}

The algorithm for \propref{oneruns} no longer attempts to align multiple traces. 
Step three is approximate by design because we use $1$-runs to fill in the gaps between the long $0$-runs. The error is from the variance of how many bits of each run are deleted by the deletion channel. See \apndref{A} for the proof.

\section{Approximating more general classes of strings}
\seclab{general-ub}\seclab{upper-bounds}

Moving on from our warm-ups, we reconstruct larger classes of strings. Our first two algorithms in this section reconstruct strings that still contain some long runs, where these help us align traces. Our third algorithm reconstructs from a single trace by approximately preserving local density.

\subsection{Identifying long runs}

To weaken the assumptions of \propref{longruns}, we want to consider strings where $0$-runs can be any length but $1$-runs must still be long and have length $\Omega(\log n)$. 
When relaxing the length restriction on the $0$-runs, the alignment step, step 1, of the algorithm for \propref{longruns}  begins to fail---entire runs of $0$s may be deleted, combining consecutive $1$-runs and making it difficult to identify which runs align together between traces. 
To still use an alignment algorithm that averages run lengths, we impose the weaker condition on the $0$-runs that they must be divided into short $0$-runs and long $0$-runs. 
As long as there is a gap of sufficiently large size such that there are no $0$-runs with length in the gap,
then in the traces we can identify which $0$-runs are long and which are short.

\gapone*

\subsubsection*{Algorithm for identifying long runs}

\begin{enumerate}
    \item[{\bf Set-up:}] String $X$ on $n$ bits such that all of its $1$-runs have length at least $C' \log(n)/ \varepsilon$, where $C' \geq 100 / p$, and all of its $0$-runs have length either greater than $3C' \log n$ or less than $C' \log n $.
    \item Sample $T = \frac{2}{p^2 \epsilon^2}\log(n)$ traces, $\widetilde{X}_1,\ldots, \widetilde{X}_T$, from the deletion channel with probability $q$. 
    \item Define $L:= 2C' p\log n$, and for all $j \in [T]$, index the $0$-runs in $\widetilde{X}_j$ with length at least $L$  as $\widetilde{\mathbf{r}}^j_1,\ldots, \widetilde{\mathbf{r}}^{j}_{k_j}$. 
    For $i \in [k_{j}-1]$, let $\widetilde{\mathbf{s}}^j_i$ be the bits between $\widetilde{\mathbf{r}}^j_{i}$ and  $\widetilde{\mathbf{r}}^j_{i+1}$ in $\widetilde{X}_j$ and let $\widetilde{\mathbf{s}}^j_0$ be the bits before $\widetilde{\mathbf{r}}^j_{1}$ and
    $\widetilde{\mathbf{s}}^j_{k_{j}+1}$ the bits after $\widetilde{\mathbf{r}}^j_{k_{j}}$ for all $j \in [T]$.
    \item If there exist $j \neq j' \in [T]$ such that $k_{j} \neq k_{j'}$, then fail without output. Otherwise, let $k:= k_{1} = k_{2} = \cdots = k_{T}$. 
    \item Compute $\widetilde{\mu}^{\mathbf{r}}_i = \frac{1}{T}\sum_{j=1}^T |\widetilde{\mathbf{r}}^j_i|$ for all $i \in [k]$ and  $\widetilde{\mu}^{\mathbf{s}}_i = \frac{1}{T}\sum_{j=1}^T |\widetilde{\mathbf{s}}^j_i|$  for all $i \in \{0\} \cup[k+1]$.
    \item Output $\widehat{X} = \widehat{1}_0\widehat{0}_1 \widehat{1}_1\cdots \widehat{1}_k\widehat{0}_k \widehat{1}_{k+1}$, where $\widehat{1}_i$ is a $1$-run, length $\frac{\widetilde{\mu}^{\mathbf{s}}_i}{p}$, and $\widehat{0}_i$ is a $0$-run, length $\frac{\widetilde{\mu}^{\mathbf{r}}_i}{p}$.
\end{enumerate}

Observe that the algorithm is inherently approximate, as we fill in the gaps between the long 0-runs with $1$-runs, omitting any short $0$-runs.

\subsubsection*{Analysis}

\begin{proof}[Proof of {\thmref{gapone}}]
 Let $X$ be a string on $n$ bits such that all of its $1$-runs have length at least $C' \log(n)/ \varepsilon$, where $C' \geq 100/ p$, and all of its $0$-runs have length either greater than $3C' \log n$ or less than $C' \log n $.
Take $T = \frac{2}{p^2 \epsilon^2}\log(n)$ traces of $X$.  By a Chernoff bound, with probability at least $1-\frac{1}{n^2}$, no $1$-run is fully deleted in any trace; in the following we assume that we are on this event.

    We will justify that in the traces we can identify all $0$-runs that had length at least $3C'\log(n)$ in~$X$.
    Let $\mathbf{r}$ be a $0$-run from $X$ with length $|\mathbf{r}| \ge 3C' \log(n)$. 
    Using a Chernoff bound, the probability that in a single trace 
    $\mathbf{r}$ is transformed into a run $\widetilde{\mathbf{r}}$ with $|\widetilde{\mathbf{r}}| \leq 2C' p \log(n)$ is bounded by %
    \begin{align*}
        \P\big(|\widetilde{\mathbf{r}}| \le 2 C' p \log(n) \big ) &  \ \le\  \P\big( ||\widetilde{\mathbf{r}}| - p |\mathbf{r}||\ge   C' p\log(n)\big)
        \le 2n^{-3}
    \end{align*}
    Similarly, for any $0$-run $\mathbf{r}$ in $X$ such that $|\mathbf{r}| \le C'  \log(n)$, 
    the probability that $\mathbf{r}$ is reduced to a run $\widetilde{\mathbf{r}}$ with $ |\widetilde{\mathbf{r}} | \geq 2C' p \log(n)$ 
    is bounded by
    $$
        \P\big(|\widetilde{\mathbf{r}}| \ge 2C' p \log(n) \big) 
         \ \le\  \P\big( ||\widetilde{\mathbf{r}}| - p |\mathbf{r}| | \ \geq\  C' p \log n \big)
        \ \le\ 2n^{-3}
  $$
    
    It follows that, with probability at least
     $1-\frac{4T}{n^2}$, there does not exist any $0$-run and any trace such that either of the ``unlikely'' inequalities above holds. 
On this event, we have that for any $0$-run $\mathbf{r}$ of length at least $3C' \log n$, and any trace $\widetilde{X}_{j}$, we can identify the image $\widetilde{\mathbf{r}}^j$ of $\mathbf{r}$ in trace $\widetilde{X}_j$. 
In particular, on this event, 
the number of $0$-runs in each trace that has length at least $2C' p \log(n)$ is equal to the number of $0$-runs in $X$ of length at least $3C'\log(n)$; 
thus $k_{1} = k_{2} = \cdots k_{T} =: k$. 
The algorithm and proof now proceed very similarly to those of \propref{oneruns}, except since we have more than a single trace, we estimate lengths of subsequences by scaling an average of the corresponding subsequences from the traces.
    
    Let $L:= 2C' p\log n  $ 
    and find every $0$-run in $\widetilde{X}_j$ with length at least $L$, indexing them as $\widetilde{\mathbf{r}}^j_1,\ldots, \widetilde{\mathbf{r}}^j_k$. 
    For $i \in [k-1]$, let $\widetilde{\mathbf{s}}^j_i$ be the bits between the last bit of $\widetilde{\mathbf{r}}^j_{i}$ and the first bit of $\widetilde{\mathbf{r}}^j_{i+1}$ in $\widetilde{X}_j$ and let $\widetilde{\mathbf{s}}^j_0$ be the bits before $\widetilde{\mathbf{r}}^j_{1}$ and
    $\widetilde{\mathbf{s}}^j_{k+1}$ the bits after $\widetilde{\mathbf{r}}^j_{k}$.
    Let ${\mathbf{s}}_i$ be the contiguous substring of $X$ from which $\widetilde{\mathbf{s}}^1_i,\ldots,\widetilde{\mathbf{s}}^T_i$ came and ${\mathbf{r}}_i$ the contiguous substring of $X$ from which $\widetilde{\mathbf{r}}^1_i,\ldots,\widetilde{\mathbf{r}}^T_i$ came. 
    
    For all $i$, we will approximate ${\mathbf{r}}_i$ with $\widehat{0}_i$  a $0$-run of length $\widetilde{\mu}^{\mathbf{r}}_i /p$ , for $\widetilde{\mu}^{\mathbf{r}}_i = \frac{1}{T} \sum_{j=1}^T |\widetilde{\mathbf{r}}^j_i|$, and we will approximate ${\mathbf{s}}_i$ with $\widehat{1}_i$, a $1$-run of length $\widetilde{\mu}^{\mathbf{s}}_i /p$, for  $\widetilde{\mu}^{\mathbf{s}}_i = \frac{1}{T} \sum_{j=1}^T |\widetilde{\mathbf{s}}^j_i|$.
    Applying a Chernoff bound and then a union bound, 
    $\P(\exists i : |\widetilde{\mu}^{\mathbf{r}}_i/p-|{\mathbf{r}}_i|| \ge \epsilon |{\mathbf{r}}_i|)  \leq 2n^{-3}$ and $\P(\exists i : |\widetilde{\mu}^{\mathbf{s}}_i/p-|{\mathbf{s}}_i|| \ge \epsilon |{\mathbf{s}}_i|)  \leq 2n^{-3}$.
    
    Since ${\mathbf{s}}_i$ contains alternating $1$-runs with length at least $C'\log(n) / \varepsilon$ and $0$-runs with length at most $C'\log(n)$,
    ${\mathbf{s}}_i$ has at least a $1-\epsilon$ density of $1$s. Therefore $\de({\mathbf{s}}_i, \widehat{1}_i) \leq 2 \varepsilon |{\mathbf{s}}_i|$ and $\de({\mathbf{r}}_i, \widehat{0}_i) \leq  \varepsilon |{\mathbf{r}}_i|$.
    Let $\widehat{X} = \widehat{1}_0\widehat{0}_1\widehat{1}_1 \cdots \widehat{1}_k\widehat{0}_k\widehat{1}_{k+1}$ and we see that from \lemref{substringrecon}
    \begin{align*}
    \de(X,\widehat{X}) &= \sum_{i=1}^k (\de(\widehat{0}_i, {\mathbf{r}}_i) + \de(\widehat{1}_i, {\mathbf{s}}_i) ) + \de(\widehat{0}_0, {\mathbf{s}}_0) +\de(\widehat{0}_{k+1}, {\mathbf{s}}_{k+1}) \\
    &\leq \sum_{i=1}^k ( \varepsilon |{\mathbf{r}}_i| + 2\varepsilon |{\mathbf{s}}_i|) + 2\varepsilon |{\mathbf{s}}_0|+2\varepsilon |{\mathbf{s}}_{k+1}| \leq 2 \varepsilon n.
    \end{align*}    
    If we apply this algorithm and analysis with $\epsilon/2$ instead of $\epsilon$, the result follows. Constants were taken large enough to account for this factor of 2. 
    \end{proof}

Note that the above theorem holds when the constant $C'$ is unknown. Given $T = O(\log n / \varepsilon^2)$ traces of $X$, we can determine whether or not $X$ had such a gap, and the corresponding $C'$ value, with high probability. We can then execute the algorithm as stated.

\subsection{Identifying dense substrings}

Here we extend the class of strings we can approximately reconstruct, proving a robust version of \thmref{gapone}. 
Specifically, we consider strings with similar properties to those in \thmref{gapone}, allowing for additional bit flips. 

\gaponerobust*

The general goal of the algorithm is similar to that of \thmref{gapone}, 
which is to identify long $0$-runs from $Y$ in each trace of $X$ and to align by these $0$-runs; then, we approximate the rest of $X$ with $1$-runs. 
Because $X$ and $Y$ have small edit distance, a good approximation for $Y$ is also good for $X$. 
Unfortunately the long $0$-runs from $Y$ are no longer necessarily $0$-runs in $X$, and therefore they are more difficult to find in the traces. 
Instead we find long $0$-dense substrings in $X$.

Let $X$ and $Y$ be as in the theorem statement. 
We also fix $m := C' \epsilon \log(n)$ throughout this subsection. 
Fix a trace $\widetilde{X}$ of $X$, as well as an index $\ell$.
Let $\widetilde{n}$ denote the length of the trace.
Define the indices $i_\ell$ and $j_\ell$ to be those that are $(m+1)$ $1$s to the left and right of $\ell$ in~$\widetilde{X}$, respectively, if such indices exist. 
We count the $0$s in $\widetilde{X}$ between indices $i_\ell$ and $j_\ell$ with the quantity 
$$ S_{\text{int}}(\widetilde{X},\ell) := \sum_{k = i_\ell}^{j_\ell} \mathds{1}_{\widetilde{X}[k] = 0}.$$ 

Note that $S_{\text{int}}(\widetilde{X},\ell) $ is not defined if $i_{\ell}$ or $j_{\ell}$ are not defined.
We use a slightly different quantity on the boundary of the trace to handle this.
Letting the definition of $i_{\ell}$ and $j_\ell$ remain the same, if $i_{\ell}$ or $j_{\ell}$ is not defined, then we consider
$ S_{\text{L-bound}}(\widetilde{X},\ell) := \sum_{k =0}^{j_\ell} \mathds{1}_{\widetilde{X}[k] = 0}$ or 
$ S_{\text{R-bound}}(\widetilde{X},\ell) := \sum_{k=i_{\ell} }^{\widetilde{n}} \mathds{1}_{\widetilde{X}[k] = 0}$, respectively. 
Combining the interior and boundary quantities, let
$S(\widetilde{X}_j,\ell) = S_{\text{int}}(\widetilde{X}_j,\ell)$ if there are $(m+1)$ 1s to the left and right of $\ell$, let $S(\widetilde{X}_j,\ell) = S_{L-\text{bound}}(\widetilde{X}_j,\ell)$ if there are $(m+1)$ 1s to the right of $\ell$ but not the left, and 
 let $S(\widetilde{X}_j,\ell) = S_{R-\text{bound}}(\widetilde{X}_j,\ell)$ if there are $(m+1)$ 1s to the left of $\ell$ but not the right.


In each trace we identify a set of substrings of $X$ that are $0$-dense, and then decide whether each such substring is long or short using $S(\widetilde{X}_j,\ell)$; 
that is, whether the corresponding unknown $0$-runs in $Y$ are long (length at least the upper bound of the gap) or short (length at most the lower bound of the gap).
If the traces all agree on the number of long $0$-dense substrings, we align the traces by these substrings and reconstruct in a manner similar to that of \thmref{gapone}. 

\subsubsection*{Algorithm for identifying dense substrings}

\begin{enumerate}
    \item[{\bf Set-up:}] String $X$ on $n$ bits formed by  flipping  at  most $ \varepsilon C' \log(n)$ bits in each run of $Y$, where $Y$ is  a string on $n$ bits such that all of its $1$-runs have length at least $C' \log(n)/\varepsilon$, for $C' \geq 1000/ p$, and all of its $0$-runs have length either greater than $3C' \log n$ or less than $C' \log n $. 
    \item Sample $T =\frac{2}{p^2 \varepsilon^2 } \log n$ traces, $\widetilde{X}_1,\ldots, \widetilde{X}_T$, from the deletion channel with deletion probability~$q$.
    \item  Set $m := \varepsilon C' \log n$ and $a := p C' \log n $. 
 For each trace $\widetilde{X}_j$, let $i$ be the smallest index of $\widetilde{X}_j$ such that $\widetilde{X}_j[i] = 0$ and
$|\{k : \widetilde{X}_j[k] = 0, |i - k| \le a+m\}| \ge a.$
Let $\ell^j_1$ be the smallest index such that $\widetilde{X}_j[\ell^j_1] = 0$ and 
$ |\{k : \widetilde{X}_j[k] = 0, i-(a+m) \le k < \ell^j_1 \}|=m.$  Compute $S(\widetilde{X}_j,\ell^j_1)$. 
Starting $m+1$ bits to the right of the last bit counted in  $S(\widetilde{X}_j,\ell^j_1)$,
continue scanning to the right and repeat this process, finding indices $\ell_t^j$ and computing $S(\widetilde{X}_j,\ell^j_t)$, for $t \geq 2$.

\item Set $\bar{G} =  2C'p \log n  $. For every trace $\widetilde{X}_j$, let
$I_j = \{ t : S(\widetilde{X}_j, \ell_t^j) > \bar{G}\}$.
If $|I_j|$ is not the same across all $T$ traces, the algorithm fails.
Otherwise, define $I= |I_j|$
and for all $t \in [I]$, we let $\widehat{0}_t$ be a $0$-run of length $\widetilde{\mu}_t /p$, 
for $\widetilde{\mu}_t = \frac{1}{T} \sum_{j=1}^T S(\widetilde{X}_j,\ell^j_t) $. 
    \item 
    Define $\widehat{i}_t = \frac{1}{T} \sum_{j'=1}^T i_{\ell^{j'}_t}$ and $\widehat{j}_t = \frac{1}{T} \sum_{j'=1}^T j_{\ell^{j'}_t}$, 
    for $i_{\ell^{j'}_t}$ and $j_{\ell^{j'}_t}$
     as in the definition of $S(\widetilde{X}_{j'},\ell^{j'}_t)$. 
    Let $\widehat{1}_0,\ldots,\widehat{1}_{I}$ be $1$-runs where $\widehat{1}_t$ has length $ | \widehat{i}_{t+1} - \widehat{j}_t|/p$ 
    for $t \in [I-1]$, $\widehat{1}_0$ has length $\widehat{i}_1/p$, and $\widehat{1}_{I}$ has length $|p n-\widehat{j}_{I}|/p$.
    \item Output $\widehat{X} = \widehat{1}_0\widehat{0}_1 \widehat{1}_1\cdots \widehat{1}_{I-1}\widehat{0}_{I-1} \widehat{1}_{I}$.
\end{enumerate}

\subsubsection*{Analysis}

    Let $\epsilon, p$ be fixed such that $p > 3\epsilon$, 
    and let $C'$ be fixed such that $C' \ge \frac{1000}{p}$. 
    Suppose $Y$ is a string on $n$ bits such that every $1$-run in $Y$ has length at least $C' \log(n) / \varepsilon$ and all of its $0$-runs have length either greater than $3C' \log n$ or length less than $C' \log n$. 
    Let $X$ be a string on $n$ bits that is formed by flipping at most $m = C' \epsilon \log(n)$ bits within each run of $Y$. Let
    $\widetilde{X}$ be a trace of $X$.
    A $0$-run $\mathbf{r}$ in $Y$ may have some bits flipped from 0 to 1 in $X$, becoming the substring $\mathbf{r}_X$, 
    so let $|\mathbf{r}_X^0|$ denote the number of $0$s in $\mathbf{r}_X$.  

Next, we prove several properties of $S(\widetilde{X},\ell)$ when the bit at index $\ell$ in trace $\widetilde{X}$ was from a $0$-run in $Y$ and $X$.

\begin{lemma}\lemlab{L-properties}
    Let $\widetilde{X}$ be a random trace from $X$, 
    and let $\ell$ be an index of $\widetilde{X}$ such that $\widetilde{X}[\ell]=0$. 
    If the bit at $\widetilde{X}[\ell]$ is from a 0-run $\mathbf{r}$ in $Y$, then the following holds for the quantity $S(\widetilde{X},\ell)$:
\begin{enumerate}
    \item   \propertylab{ijproperty}\label{ijproperty} (Property 1) With probability at least $1-n^{-6}$ 
    the bits at indices $i_\ell$ and $j_{\ell}$ come from a $1$-run adjacent to $\mathbf{r}$. 
    \item  \propertylab{distlemma}\label{distlemma} (Property 2)  If indices $i_\ell$ and $j_{\ell}$ come from a $1$-run adjacent to $\mathbf{r}$, then 
    $S(\widetilde{X},\ell)$ is upper bounded by a random variable from the distribution $\mathrm{Bin}(|\mathbf{r}_X^0|,p) + \mathrm{Bin}(2m,p)$.
     \item \propertylab{lowvarexp}\label{lowvarexp} (Property 3)
    If $|\mathbf{r}| \geq C' \log n$ and the bits at indices $i_\ell$ and $j_{\ell}$ come from a $1$-run adjacent to $\mathbf{r}$, then with probability at least $1-n^{-6}$, 
    $|S(\widetilde{X},\ell) - p|\mathbf{r}|| \le \frac{p |\mathbf{r}|}{4} + 3m.$
\end{enumerate}
\end{lemma}

\begin{proof}[Proof of Property 1]
    It suffices to prove the claim for $i_\ell$. 
    Index $i_\ell$ is $m+1$ $1$s to the left of $\ell$, and therefore not from $\mathbf{r}$, since at most $m$ $0$s of $\mathbf{r}$ were flipped to $1$s. 
    Further, by a Chernoff bound, with probability at least $1-n^{-6}$ the $1$-run left-adjacent to $\mathbf{r}$ in $Y$ has at least $2m+1$ bits surviving in $\widetilde{X}$. 
    At most $m$ bits of the left-adjacent $1$-run to $\mathbf{r}$ in $Y$ are flipped to $0$, 
    so at least $m+1$ $1$s from this $1$-run survive in $\widetilde{X}$.
    It follows that $i_\ell$ came from the left adjacent $1$-run to $\mathbf{r}$ in $Y$.
\end{proof}
\begin{proof}[Proof of Property 2]
    Recall that $|\mathbf{r}_X^0|$ is the number of $0$s in $\mathbf{r}$ that were not flipped to $1$ in $X$.
    This component of $S(\widetilde{X},\ell)$ is from the distribution $\mathrm{Bin}(|\mathbf{r}_X^0|, p)$. 
    Let the contribution to $S(\widetilde{X},\ell)$ by any $0$s not from $\mathbf{r}$ be the random variable $Z_{\mathbf{r}}(\ell)$. 
    Each bit that was flipped to $0$ in either $1$-run adjacent to~$\mathbf{r}$ in $Y$ can contribute $1$ with probability at most $p$ to $Z_{\mathbf{r}}(\ell)$. 
    From the assumption on $i_{\ell}$ and $j_\ell$, any other $0$ from $X$ will be outside of the range $[i_\ell, j_\ell]$. 
    Therefore we can upper bound the contribution of  $Z_{\mathbf{r}}(\ell)$ by a random variable sampled from $\mathrm{Bin}(2m,p)$.
\end{proof}
\begin{proof}[Proof of Property 3]
     By \propertyref{distlemma}, $S(\widetilde{X},\ell)$ is upper bounded by a random variable from the distribution 
     $\mathrm{Bin}(|\mathbf{r}_X^0|,p) + Z_{\mathbf{r}}(\ell)$. 
     By a Chernoff bound, with probability $1-n^{-6}$ the first binomial term varies from its mean by at most $p |\mathbf{r}|/4$. 
     The second binomial term is upper bounded by $2m$ and $||\mathbf{r}_X^0| - |\mathbf{r}|| \le m$.
\end{proof}







Now we are ready to prove \thmref{gaponerobust}.
\begin{proof}[Proof of {\thmref{gaponerobust}}]
Define $a := p C'\log(n)$. Take $ T= \frac{2}{p^2 \varepsilon^2} \log n $ traces of $X$, $\widetilde{X}_1,\ldots, \widetilde{X}_T$ , and
fix a trace $\widetilde{X}_j$. 
Our first goal is to find long $0$-dense substrings in $X$; we can also think of these long $0$-dense substrings as corresponding to long $0$-runs in $Y$. 
Let $i$ be the smallest index of $\widetilde{X}_j$ such that $\widetilde{X}_j[i] = 0$ and there are at least $a$ 0s in $\widetilde{X}_j$ within $a+m$ indices of $i$, i.e.
$$|\{k : \widetilde{X}_j[k] = 0, |i - k| \le a+m\}| \ge a.$$
Next find the index $\ell^j_1$ such that $\widetilde{X}_j[\ell^j_1] = 0$ and there are exactly $m$ 0s in $\widetilde{X}_j$ within the interval of indices $[ i - (a+m),\ell^j_1]$, i.e.
$|\{k : \widetilde{X}_j[k] = 0,\ i-(a+m) \le k < \ell^j_1 \}| = m.$ 
The goal of this procedure is to find an index $\ell^j_1$ such that the bit at $\widetilde{X}_j[\ell^j_1]$ 
is from a $0$-run in $Y$ with high probability.

With probability at least $1-n^{-6}$,  
every $1$-run in $Y$ is reduced to a substring with at least $2(a+m)$ $1$s in $\widetilde{X}_j$. This implies that the length $2(a+m)$ interval $\widetilde{X}_j[i-(a+m), i+a+m]$ contains bits from at most two $1$ runs in $Y$ and at most one $0$ run with probability $1-n^{-6}$. By construction, this interval contains at least $a > 3m$ $0$s (the inequality coming from the fact that $p > 3\epsilon$). Since each $1$-run had at most $m$ bits flipped to $0$, there must be at least $a-2m > m$ $0$s in the interval $\widetilde{X}_j[i-(a+m), i+a+m]$ that came from some $0$-run $\mathbf{r}$ in $Y$. In this construction, the $0$s from the $\mathbf{r}$ that survived in $\widetilde{X}_j$ are nested between at most $m$ $0$s that were flipped from the left-adjacent $1$-run to $\mathbf{r}$ in $Y$ and at most $m$ $0$s that were flipped from the right-adjacent $1$-run to $\mathbf{r}$ in $Y$. This implies that the $(m+1)$th $0$ in this interval must be from the $0$-run $\mathbf{r}$.

Compute $S(\widetilde{X}_j,\ell^j_1)$. Note that with high probability, if a trace does not have $(m+1)$ 1s to the right of $\ell^j_1$, the original string can be well-approximated by outputting the all 0s string with length $\frac{1}{T}\sum_{j=1}^T|\widetilde{X}_j|/p$. 
Starting $m+1$ bits to the right of the last bit counted in  $S(\widetilde{X}_j,\ell^j_1)$,
continue scanning to the right and repeat this process, finding indices $\ell_t^j$ and computing $S(\widetilde{X}_j,\ell^j_t)$,
for $t \geq 2$. 
We jump ahead $m+1$ bits to the right between iterations because this forces the next bit $i$ that satisfies the condition 
$|\{k : \widetilde{X}_j[k] = 0, |i - k| \le a+m\}| \ge a$ to not overlap with the previous $0$-run with high probability by \propertyref{ijproperty}.



We justify that this process succeeds, meaning that it catches all long $0$-runs from $Y$, in all $T$ traces, with high probability.
For $0$-run $\mathbf{r}$ in $Y$ such that $|\mathbf{r}| \ge 3C' \log(n)$,
with probability at least $1-n^{-6}$ 
at least $a+m$ bits from all such $0$-runs survive in all $T$ traces. 
Further there are at most $m$ $1$s among these bits. Therefore, with probability at least $1-n^{-6}$, 
we have at least $a$ $0$s that have at most $m$ $1$s inserted among them, and this triggers the calculation of $\ell^j_t$ for some $t$.


By the theorem assumptions, there exists an interval $[C' \log n, 3 C' \log n]$ such that no $0$-run $\mathbf{r}$ in $Y$ has $|\mathbf{r}|$ in the gap $[C' \log n, 3C' \log n]$. 
Let $\bar{G}$ be the middle of the gap scaled by $p$, so $\bar{G} =  2C'p \log n   $.
By \propertyref{lowvarexp} and a union bound, with probability at least $1-n^{-4}$, 
all $0$-runs $\mathbf{r}$ in $Y$ with $|\mathbf{r}| \geq 3C' \log n$ will trigger the calculation of an $\ell^j_t$ with $S(\widetilde{X}_j,\ell^j_t) > \bar{G}$
 in all traces,
and all $0$-runs $\mathbf{r}$ in $Y$ with $|\mathbf{r}| < C' \log n $ will either not trigger an $\ell^j_t$ calculation, or if they do,
 $\ell^j_t$ will have $S(\widetilde{X}_j,\ell^j_t)  < \bar{G}$ for all traces.

 For every trace $\widetilde{X}_j$, let $I_j = \{t :S(\widetilde{X}_j,\ell^j_t)> \bar{G} \}$.
If $|I_j|$ is not the same across all $T$ traces, the algorithm fails.
Otherwise let $I=|I_j|$ for all $j$, and for each trace $\widetilde{X}_j$ relabel the $\ell^j_t$ with $S(\widetilde{X}_j,\ell^j_t)> \bar{G}$ as  
$\ell_1^j, \ldots, \ell_I^j$.

The proof now proceeds similarly to that of \thmref{gapone}.
We approximate long $0$-runs $\mathbf{r}_t$ in $Y$, 
which are close to some long $0$-dense substrings of $X$ with high probability, with $0$-runs, and the rest is approximated with $1$-runs.
We first estimate the distance between the $0$-runs in $Y$. 
Consider a $0$-run $\mathbf{r}_t$ that generates an estimate of $\widetilde{\mu}^{\mathbf{r}}_t /p$, 
and take $\widehat{i}_t = \frac{1}{T} \sum_{j'=1}^T i_{\ell_t^{j'}}$ and $\widehat{j}_t = \frac{1}{T} \sum_{j'=1}^T j_{\ell_t^{j'}}$, for $i_{\ell_t^{j'}}$ and $j_{\ell_t^{j'}}$ 
as in the definition of $S(\widetilde{X}_{j'},\ell_t^{j'})$. The average of the indices
$\widehat{i}_t$ can be at most $m$ bits to the left of the first $0$ from $\mathbf{r}_t$, and therefore is at most off by $m$ bits. 
The same is true for $\widehat{j}_t$. 
By a Chernoff bound, $ | \widehat{i}_{t+1} - \widehat{j}_t|/p$  is an estimate of the distance between $0$-runs with accuracy $2\epsilon | \mathbf{r}_t|$ with probability at least $1-n^{-6}$. 
The substring between these $0$-runs also has at least a $1-\epsilon$ density of $1$s, so we can fill with $1$-runs for a good approximation. 
Let $\widehat{1}_0,\ldots,\widehat{1}_{I}$ be $1$-runs where $\widehat{1}_t$ has length $ | \widehat{i}_{t+1} - \widehat{j}_t|/p$ for $t \in [I-1]$, $\widehat{1}_0$ has length $\widehat{i}_1/p$, and $\widehat{1}_{I}$ has length $|p n-\widehat{j}_{I}|/p$. 
Hence by \lemref{substringrecon} the $1$-runs contribute at most $3 \epsilon n$ to the edit distance error.

It remains to estimate the lengths of the long $0$-runs in $Y$ $\mathbf{r}_1,\ldots,\mathbf{r}_I$.
Fix $t \in [I]$, let $\widehat{0}_t$ be a $0$-run of length $\widetilde{\mu}^{\mathbf{r}}_t /p$,
 for $\widetilde{\mu}^{\mathbf{r}}_t = \frac{1}{T} \sum_{j=1}^T S(\widetilde{X}_j,\ell^j_t)$.
 For every $\mathbf{r}_t \in \{\mathbf{r}_1,\ldots,\mathbf{r}_I\}$, define ${\mathbf{r}_t}_X^0$ as above (the number of $0$s from $\mathbf{r}_t$ in $X$). With probability at least $1-n^{-6}$ 
 the average of $\mathrm{Bin}(|{\mathbf{r}_t}_{X}^0|, p)$
over $T=O(\log(n)/\epsilon^2)$ traces is within $\epsilon p|{\mathbf{r}_t}_{X}^0|$ of the mean $p|{\mathbf{r}_t}_{X}^0|$. 
Combining this with \propertyref{distlemma}, with probability at least $1-n^{-3}$,
$$|\widetilde{\mu}^{\mathbf{r}}_t - p|{\mathbf{r}_t}_X^0|| \le \epsilon p|{\mathbf{r}_t}_X^0| + 2m.$$
Since $|{|\mathbf{r}_t}_X^0| - |\mathbf{\mathbf{r}_t}|| \le m$,
we have that 
$$|p|\mathbf{\mathbf{r}_t}| - \widetilde{\mu}^{\mathbf{r}}_t | \le \epsilon p |\mathbf{\mathbf{r}}_t| + 2m + pm = \epsilon p|\mathbf{\mathbf{r}}_t| + 3m.$$
This is at worst an approximation of $p|\mathbf{\mathbf{r}}_t|$ with edit distance error  at most
$$\epsilon + \frac{3m}{p|\mathbf{\mathbf{r}}_t|} \le \epsilon + \frac{3m}{p(a-2m)} \le \epsilon + \frac{9\epsilon}{p^2} \le C''\epsilon,$$
where we use $a > 3m$ and $C'' = 1 + \frac{9}{p^2}$.
Taking a union bound over all ${\mathbf{r}}_t \in \{\mathbf{r}_1, \ldots,\mathbf{r}_I\}$, and applying \lemref{substringrecon}, with probability at least $1-n^{-2}$ the long $0$-run estimates contribute at most error $C''\epsilon n$. Putting this all together,  we output the string $\widehat{X} = \widehat{1}_0\widehat{0}_1 \widehat{1}_1 \cdots \widehat{1}_{I-1}\widehat{0}_{I-1} \widehat{1}_{I}$.
One more application of \lemref{substringrecon} implies that $\de(Y,\widehat{X}) \leq (C'' +3) \varepsilon n$.
Since $Y$ is within $\epsilon n$ edit distance from $X$, the triangle inequality lets us conclude that
$\de(X,\widehat{X}) \leq (C'' +4)\epsilon n $.

    If we apply this algorithm and analysis with $\frac{\epsilon}{ C''+4}$ instead of $\epsilon$, the result follows. 
    Constants were taken large enough to account for this factor of $C''+4$.
\end{proof}

As before, the theorem holds when the constant $C'$ is unknown. Given $T = O(\log n / \varepsilon^2)$ traces of $X$, we can find whether $X$ has a gap, and the corresponding $C'$ value, with high probability.

\subsection{Majority voting in substrings}

A natural follow-up question to the previous theorems is what happens when the string no longer has long runs, but instead has long dense regions. 

\singletracetwo*

\subsubsection*{Algorithm for majority voting in substrings}

\begin{enumerate}
\item[{\bf Set-up:}] String $X$ on $n$ bits such that $X$ can be divided into contiguous intervals all of length at least $L=50 \log n / (p^2\varepsilon^2)$ and density at least $1-\frac{\epsilon}{12}$ of $0$s or $1$s.
\item Sample a single trace $\widetilde{X}$ from the deletion channel with probability $q$.
    \item Uniformly partition $\widetilde{X}$ into contiguous substrings of length $w =\epsilon p L$, so $\widetilde{X} = \widetilde{X}_1\cdots \widetilde{X}_{\lceil n/w \rceil}$, with a shorter last interval if needed. 
    \item Output $\widehat{X} = \widehat{X}_1 \cdots \widehat{X}_{\lceil n/w \rceil}$, where $\widehat{X}_i$ is a run of length $w/p$ with value the majority bit of $\widetilde{X}_i$ for $i \in [\lceil n/w \rceil]$. 
\end{enumerate}

\subsubsection*{Analysis}

We first present three properties of the traces generated by high density strings with large length.

\begin{lemma}
\lemlab{highdensity1}

Fix $\epsilon$ and $p$. Let $X$ be a string on at least $L$ bits, where $L = \frac{50}{p^2\epsilon^2}\log(n)$
    with density of at least $1-\epsilon$ of either $0$ or $1$. For a trace $\widetilde{X}$ of $X$, the following properties hold with probability at least $1-n^{-4}$.
\begin{enumerate}
    \item \propertylab{atleasthalf}\label{atleasthalf}(Property 1)\quad $\frac{|\widetilde{X}|}{L} \ge \frac{p}{2}$

    \item  \propertylab{lengthhighdensity}\label{lengthhighdensity}(Property 2)\quad  $\Big||X| - \frac{|\widetilde{X}|}{p}\Big| \le \epsilon |X|$.

    \item \propertylab{densityhighdensity}\label{densityhighdensity}(Property 3)\quad $\widetilde{X}$ has density at least $1-2\epsilon$ of $0$s. 

\end{enumerate}
\end{lemma}

\begin{proof}
Assume w.l.o.g. that $X$ has density at least $1-\epsilon$ of $0$. Applying a Chernoff bound gives that with probability at least $1-n^{-6}$, 
the length of $\widetilde{X}$ is in the range $p|X| \pm \sqrt{3|X|\log(n)}$. Taking this lower bound gives $\frac{|\widetilde{X}|}{|X|} \ge p - \frac{\sqrt{3|X|\log(n)}}{|X|}$.  Since $|X| \ge L$, we see that  $\frac{\sqrt{3|X|\log(n)}}{|X|} \leq p/2$, completing the proof of \propertyref{atleasthalf}. 
Another way of writing the same Chernoff bound result is that $\left ||X| - \frac{|\widetilde{X}|}{p}\right| \le \sqrt{3|X|\log(n)} \le \epsilon |X|$, proving  \propertyref{lengthhighdensity}.

Applying a Chernoff bound to the number of $0$s in $X$, 
with probability at least $1-n^{-6}$, the number of non-deleted $0$s is at least 
$p|X|(1-\epsilon) - \sqrt{3|X|(1-\epsilon)\log(n)} \ge p|X|(1-\epsilon) - \sqrt{3|X|\log(n)}$. 
Combining this with the first application of a Chernoff bound, a union bound gives that with probability at least $1-n^{-5}$, the density of $0$s in the trace 
(denoted $\rho$) satisfies the following inequalities:
    \begin{align*}
        \rho \ge \frac{p|X|(1-\epsilon) - \sqrt{3|X|\log(n)}}{p|X| + \sqrt{3|X|\log(n)}} 
          \ge \frac{50(1-\epsilon) - \sqrt{150}\epsilon}{50 + \sqrt{150}\epsilon} 
        \ge 1 - 2\epsilon. 
    \end{align*}
Note that the second inequality comes from the fact that the expression to the left is increasing in $|X|$, and therefore is minimized at $|X| = L$.
\end{proof}

Using these results, we can now proceed to the main proof of this section.

\begin{proof}[Proof of {\thmref{singletracetwo}}]
    Suppose $X$ is a binary string on $n$ bits that can be divided into intervals $I_1,\ldots,I_m$ such that all intervals $I_i$ have length at least $L := \frac{50}{p^2\epsilon^2}\log(n)$ and density at least $1-\epsilon$ of either $0$ or $1$. Take a trace $\widetilde{X}$. Define $w = \epsilon p L$. Divide the trace $\widetilde{X}$ into consecutive intervals of width $w$ denoted as $\widetilde{X}_1,\ldots,\widetilde{X}_k$, 
    where $\widetilde{X}_i = \widetilde{X}[(i-1)w, iw]$ (with $\widetilde{X}_k$ shorter if necessary). 
    
    Our approximate string is $\widehat{X} = \widehat{X}_1 \cdots \widehat{X}_{k}$, where $\widehat{X}_k$ is a run of length $w/p$ with value the majority bit of $\widetilde{X}_i$ for $i \in [k]$, and 
    define $X_i$ to be the range of bits in $X$ that correspond to the bits in $\widetilde{X}_i$. Define $\widetilde{I}_i$ as the bits present in $\widetilde{X}$ from the interval $I_i$ in $X$.

    Consider $I_i$ for some $i$ that w.l.o.g. has majority bit $0$. By \propertyref{densityhighdensity} of \lemref{highdensity1} , at most $2\epsilon|\widetilde{I}_i|$ bits in $\widetilde{I}_i$ are $1$. Consider all intervals $\widetilde{X}_j$ such that $\widetilde{X}_j \subset \widetilde{I}_i$. There are at least $\frac{|\widetilde{I}_i|- 2w}{w}$ such intervals $\widetilde{X}_j$. At most $\frac{2\epsilon|\widetilde{I}_i|}{\frac{w}{2}} = \frac{4\epsilon|\widetilde{I}_i|}{w}$ of these intervals $\widetilde{X}_j$ can have majority bit $1$. Therefore, the fraction of these intervals that have majority bit $1$ is upper bounded by the following for $\epsilon \le \frac{1}{8}$, where we use \propertyref{atleasthalf} of \lemref{highdensity1} to say that $|\widetilde{I_i}| \ge \frac{pL}{2}$:
    
    $$\frac{4\epsilon}{1-2\frac{w}{|\widetilde{I_i}|}} \le \frac{4\epsilon}{1-4\epsilon} \le 8\epsilon$$
    
    Thus, in the concatenation of $\frac{w}{p}$ of the majority bits of all $X_j$ such that $X_j \subset I_i$, the fraction of $1$s is at most $8\epsilon$. Furthermore, the length of this concatenation is within $\epsilon |I_i| + \frac{2|w|}{p}$ of $|I_i|$, where the first term comes from \propertyref{lengthhighdensity} in \lemref{highdensity1} and the second term comes from the two intervals $X_j$ that could cross both $I_i$ and either $I_{i-1}$ or $I_{i+1}$. This approximation of $I_i$ therefore has density at least $1-8\epsilon$ of $0$ and length differing by a fraction of $\epsilon + \frac{2|w|}{p|I_i|} \le 3\epsilon$ from $I_i$. Therefore, this is a total of a $11\epsilon$ approximation of $I_i$. This is true for all $i$.
    
    The last error that needs to be considered in our algorithm is the bits from $\widetilde{X}_j$ for all $j$ such that $X_j \not\subset I_i$ for all $i$ (in other words $X_j$ is on a boundary). We can assume that the bits in the output string from these $\widetilde{X}_j$ are all errors, and there are at most $\frac{n}{L}$ such boundaries. Therefore, this contributes a total error of $\frac{w}{p}\frac{n}{L} \le \epsilon n$ bits.
Putting it all together with \lemref{substringrecon},
$ \de(X,\widehat{X}) \leq 12\epsilon n$. 
If we apply this algorithm and analysis with $\epsilon/12$ instead of $\epsilon$, the result follows.
\end{proof}

\section{Lower bounds for approximate reconstruction}
\seclab{lower-bounds}

We turn our attention to proving limitations of approximate reconstruction. We provide two results, one for edit distance approximation and another for Hamming distance. Throughout this section we fix the deletion probability to be a constant $q = \Theta(1)$.

\subsection{Lower bound for edit distance approximation}

Let $\alpha \in (0,1/2)$ denote a fixed constant. Let $f(n')$ be a lower bound on the number of traces required to distinguish between two length $n'$ strings $X'$ and $Y'$ with probability at least $1/2 + \alpha$. We can take $\alpha$ to be as small as we like by slightly decreasing the lower bound, and therefore, we assume that $\alpha = 1/8$.  Previous work identifies two strings such that $f(n) = \widetilde{\Omega}(n^{1.5})$, where the $\widetilde{\Omega}$ hides the $1/\mathrm{polylog}(n)$ factor~\cite{Chase19, HL18}.
They use $X' = (01)^k 1 (01)^{k+1}$ and $Y' = (01)^{k+1}1(01)^k$ for $n' = 4k + 3$. Our strategy holds for any family of pairs $X', Y'$ that witness the lower bound. However, we note that this specific pair is already close in edit distance, and hence, outputting either of them would always be an approximation within edit distance two. 

We instead form a string $V$ of length $n$ by concatenating a sequence of blocks, where each {\em block} is a uniformly random choice between $X'$ and $Y'$. Setting the block length to be $C/\epsilon$, we show that any algorithm that approximates $V$ within edit distance~$\epsilon n$ must require at least $f(C/\epsilon)$ traces for a constant $C \in (0,1)$. Our strategy follows previous results on exact reconstruction lower bounds that argue based on traces being independent of the choice of string in each block~\cite{Chase19,HL18, tr-revisited}. However, the proof is not a straightforward extension because we must account for the algorithm being approximate. In essence, we argue that if the algorithm outputs a good enough approximation, then it must be able to distinguish between $X',Y'$ in many blocks.

\subsubsection*{Input Distribution and Indistinguishable Blocks}
We define the hard distribution as follows. Let $X'$ and $Y'$ be strings of length $1/\lceil 128 \epsilon \rceil$. We construct a random string $V$ of length $n$ by concatenating $b = \lceil 128 \epsilon \rceil n$ blocks $V = V_1V_2\cdots V_b.$ Each of the substrings $V_i$ is set to be $X'$ or $Y'$ uniformly and independently. 
The approximate reconstruction algorithm receives $T < f(C/\epsilon)$ traces for $C = 1/128$.  
By assumption, with $T$ traces from $X'$ or $Y'$, the algorithm must fail to distinguish between them with probability at least $1/2 - \alpha$. As this is an information-theoretical statement, 
we next argue that 
the $T$ traces are independent of the choice between $X'$ and $Y'$ with probability at least $1-2\alpha$. 

To formalize this claim, we introduce some notation. Let $\mathcal{A}$ denote a set of $T < f(C/\epsilon)$ traces generated from the random string $V$ described above by passing $V$ through the deletion channel $T$ times. 
Since the channel deletes bits independently, we can equivalently determine the set $\mathcal{A}$ of traces  by passing each block $V_i$ for $i \in [b]$ through the channel one at a time and then concatenating the subsequences to form a trace from $V$. We let $\mathcal{D}_i$ denote the distribution over sets of $T$ traces where $V_i$ generates these traces. By our assumption, any algorithm that receives $T < f(C/\epsilon)$ traces must fail to distinguish between $V_i = X'$ and $V_i = Y'$ with probability at least $1/2 - \alpha$. 

Next, we decompose the trace distribution $\mathcal{D}_i$ in a way that relates the failure probability to the event that the $T$ traces are independent of $V_i$. We express the distribution $\mathcal{D}_i$ over $T$ traces of $V_i$ as a convex combination of two distributions  $\mathcal{F}$ and $\mathcal{G}_{V_i}$, where intuitively sampling from $\mathcal{F}$ corresponds to being unable to determine $V_i$ with any advantage (see \defref{indistinguish} below). Formally, we take $\mathcal{F}$ and $\mathcal{G}_{V_i}$ to be any distributions over $T$ traces of $V_i$ such that for some $\gamma \in [0,1]$ we have
\begin{equation}\eqlab{convex-distribution}
\mathcal{D}_i = (1-\gamma) \cdot\mathcal{F}  + \gamma \cdot \mathcal{G}_{V_i},
\end{equation}
where $\mathcal{G}_{V_i} = \frac{1}{2}(\mathcal{G}_{X'} + \mathcal{G}_{Y'})$,
and moreover, the following three properties hold: (i) $\mathcal{F}$ is independent of $V_i$, (ii) $\mathcal{G}_{V_i}$ is not independent of whether $V_i=X'$ or $V_i = Y'$, and (iii) the distributions $\mathcal{G}_{X'}$ and $\mathcal{G}_{Y'}$ have disjoint supports.  We sketch how to construct distributions as in \Eqref{convex-distribution}. The distribution $\mathcal{D}_i$ from $V_i$ is discrete over the $T$-wise product of distributions over $\{0,1\}^{\leq n}$. Depending on $V_i$, the distribution gives different weights to each subsequence based on its length and the number of times it is a subsequence of $V_i$. 
Assume that some $T$ traces have higher probability of occurring under $X'$ than $Y'$. Assign the mass in $\mathcal{D}_i$ that comes from $Y'$ to $\mathcal{F}$ and the remainder to $\mathcal{G}_{X'}$ (if the probability is higher for $Y'$, swap $X'$ and $Y'$). Doing this for all multisets of $T$ subsequences leads to $\mathcal{G}_{X'}$ and $\mathcal{G}_{Y'}$ having disjoint support. The parameter $\gamma$ normalizes the distributions.

We now argue that $\gamma \leq 2\alpha$ by claiming that there is an algorithm using $T$ traces with failure probability at most $(1-\gamma)/2$. By our hypothesis, with $T < f(C/\epsilon)$ traces, any algorithm has failure probability at least $1/2 - \alpha$. This implies that $(1-\gamma)/2 \geq 1/2 - \alpha$, which leads to $2\alpha \geq \gamma$.  Since $\mathcal{G}_{X'}$ and $\mathcal{G}_{Y'}$ have disjoint supports, the traces from these distributions identify $V_i$, and the algorithm correctly determines $V_i$. From \Eqref{convex-distribution}, with probability $\gamma$, the traces are sampled from $\mathcal{G}_{X'}$ or $\mathcal{G}_{Y'}$. Otherwise, with probability $(1-\gamma)$, traces are sampled from $\mathcal{F}$. When traces come from $\mathcal{F}$, an algorithm that outputs either $X'$ or $Y'$ has probability $1/2$ of being correct.





Now, define a binary latent variable $\mathcal{E}_i$ such that $\mathcal{E}_i = 1$ with probability $1-\gamma$ and $\mathcal{E}_i = 0$ with probability $\gamma$. If $\mathcal{E}_i = 1$, then $\mathcal{D}_i$ samples $T$ traces from $\mathcal{F}$, and if $\mathcal{E}_i = 0$, it samples from $\mathcal{G}_{V_i}$. Using this notation, we can define the event that the traces are independent of a block in $V$. Recall that we sample $T$ traces $\mathcal{A}$ from $V$ by sampling $T$ traces from $\mathcal{D}_i$ for each $i \in [b]$ and then concatenating the traces of the blocks (using an arbitrary but fixed ordering of the traces).
\begin{defn}\deflab{indistinguish} For $i \in [b]$, we say that the $i^\mathrm{th}$ block is {\bf indistinguishable} from the $T$ traces $\mathcal{A}$ of~$V$ if the distribution $\mathcal{D}_i$ samples the traces of the $i^\mathrm{th}$ block $V_i$ from $\mathcal{F}$, or in other words, if $\mathcal{E}_i = 1$.
\end{defn}
\begin{lemma} \lemlab{indistinguishable}
	If $\alpha = 1/8$ and the number of blocks $b$ satisfies $b \geq  128 \log n$, then at least $(1-4\alpha)b$ blocks are indistinguishable with probability at least $1-2/n^2$.
\end{lemma}
\begin{proof}
Using the notation and arguments by \Eqref{convex-distribution}, we have that $\gamma \leq 2\alpha$, which implies that $\mathcal{E}_i = 1$ with probability at least $(1-2\alpha)$. Hence, the expected number of  indistinguishable blocks is at least $(1-2\alpha)b$. Since traces are generated for each block independently, the binary random variables $\{\mathcal{E}_i\}_{i=1}^b$ are independent. By a Chernoff bound, the probability that the number of indistinguishable blocks deviates from its mean by $2\alpha b$ is at most $2e^{-4\alpha^2(1-2\alpha)b/3} \leq 2 e^{-2\log n} = 2 n^{-2}$, where we have used that $(1-2\alpha) = 3/4$ and $\alpha^2 b \geq (128/64)\log n =  2 \log n$.    
\end{proof}

\subsubsection*{From Indistinguishable Blocks to Edit Distance Error}

We move on to a technical lemma that allows us to lower bound the edit distance by looking at the indicator vectors for the agreement of substrings in an optimal alignment. In what follows, we consider partitions into substrings, which are collections of non-overlapping, contiguous sequences of characters (a substring may be empty; substrings in a partition may have varying lengths). 

 \begin{lemma} 
 \lemlab{block-ed-lb}
 Let $U$ and $V$ be strings. For an integer $b \geq 1$, assume that $V$ is partitioned into $b$ substrings $V = V_1V_2\cdots V_b$. Then, there exists a partition of $U$ into $b$ substrings $U = U_1U_2\cdots U_b$ such that\footnote{
It is tempting to conjecture that equality can be achieved in \lemref{block-ed-lb} if we instead take the minimum over all partitions of $U$. However, an example shows that this does not always hold. Over the alphabet $\{\mathsf{x},\mathsf{y},\mathsf{z}\}$, consider the pair $U = \mathsf{y}\mathsf{z}\mathsf{z}\mathsf{z}\mathsf{x}$ and $V = \mathsf{x}\mathsf{y}\mathsf{y}\mathsf{x}$. Their edit distance is $\de(U,V)=4$. Using four blocks, partition $V = [\mathsf{x}][\mathsf{y}][\mathsf{y}][\mathsf{x}]$. Decompose $U = [\varnothing][\mathsf{y}][\mathsf{z}\mathsf{z}\mathsf{z}][\mathsf{x}]$. Summing the indicator vectors only equals two, and not four.
}
 $$ \de(U,V) \geq \sum_{i=1}^b \mathds{1}_{\{U_i \neq V_i\}}.$$
 \end{lemma}
 \begin{proof}
 Let $d = \de(U,V)$. We proceed by induction on the number of substrings $b$. For the base case, $b=1$, we have that the edit distance between $U$ and $V$ is zero if and only if $U = V$. For the inductive step, assume the lemma holds up to $b-1$ substrings with $b \geq 2$. We consider two cases, where in both we will split $U$ into two substrings $U = U_1U'$. 

 For the first case, assume that $V_1$ matches the prefix of $U$, so that $U = U_1U' = V_1U'$. Then, we have that $\de(U,V) = \de(U',V_2 \cdots V_b)$. Applying the inductive hypothesis with $b-1$ substrings for the pair $U'$ and $V_2 \cdots V_b$ finishes this case. 

 For the second case, $V_1$ does not match the prefix of $U$, and hence, any minimum edit distance alignment between $U$ and $V$ uses at least one edit in the $V_1$ portion. Consider any alignment between $U$ and $V$ with $d = \de(U,V)$ edits. Let $U = U_1U'$ denote the partition where $U_1$ is aligned to $V_1$ and $U'$ is aligned to $V_2\cdots V_b$. Since the prefixes differ, we have $\de(U_1,V_1) \geq 1$, which implies that $\de(U',V_2\cdots V_b) \leq d-1$. Applying the inductive hypothesis with $b-1$ substrings to the pair $U'$ and $V_2\cdots V_b$ leads to a partition $U' = U_2\cdots U_b$ such that $\sum_{i=2}^b \mathds{1}_{\{U_i \neq V_i\}} \leq d-1$. 
 We conclude that $\de(U,V) = d = 1 + (d-1) \geq \mathds{1}_{\{U_1 \neq V_1\}} + \sum_{i=2}^b \mathds{1}_{\{U_i \neq V_i\}}$ for this partition of $U$.
 \end{proof}

Using the above lemmas, we can now prove the edit distance lower bound theorem.

\blackboxlowerbound*
\begin{proof}
Let $\epsilon^\star$ be a small constant such that $\epsilon \leq \epsilon^\star < C$ and $f(C/\epsilon^\star) > 1$, where we set $C = 1/128$.  Assume that the approximate reconstruction algorithm receives $T < f(C/\epsilon)$ traces.  

Let $\widehat X$ denote the output of the reconstruction algorithm on input $V = V_1V_2\cdots V_b,$ where $V_i \in \{X', Y'\}$ and $b = \lceil 128 \epsilon \rceil n$. 
Assume for contradiction that $\de(\widehat X, V) \leq \epsilon n$ with high probability. 
Using \lemref{block-ed-lb}, we can partition $\widehat X$ into $b$ blocks $\widehat X =\widehat X_1 \widehat X_2 \cdots \widehat X_b$  such that
\begin{equation}\Eqlab{lb-edit-sum}
\de(\widehat X, V) \geq \sum_{i=1}^b \mathds{1}_{\{\widehat X_i \neq V_i\}}.
\end{equation}
Since $b \geq 128 \log n$, \lemref{indistinguishable} establishes that there are at least $(1-4\alpha)b$ blocks in $V$ that are indistinguishable with high probability using the $T$ traces. For each of these blocks, the algorithm cannot guess between $V_i = X'$ or $V_i = Y'$ with any advantage. While we do not know how the alignment corresponds to the indistinguishable blocks, we know that for at least $(1-4\alpha)b$ values $j \in [b]$, we have that $\{\widehat X_j \neq V_j\}$ with probability at least 1/2. 
Thus, the sum in \Eqref{lb-edit-sum} is at least $\frac{1}{2}(1-4\alpha)b = b/4$ in expectation, and by a Chernoff bound, it is at least $b/8$ with high probability. This implies that $\de(\widehat X, V) \geq 16\epsilon n$, contradicting the edit distance being at most $\epsilon n$.
\end{proof}

\corref{LBCorollary} now follows immediately from this theorem and the previous trace reconstruction lower bounds~\cite{Chase19}, showing that for $\delta \in (0,1/3)$, we have that
$n^{1 + 3\delta/2}/\mathrm{polylog}(n)$ 
traces are necessary to $n^{1/3 - \delta}$-approximately reconstruct an arbitrary $n$-bit string with probability $1-1/n$.

\subsection{Lower Bound for Hamming Distance Approximation}


\lowerboundHamming*
\begin{proof}
    Let $n = 4k+1$. 
   Define $X =0^k(01)^k0^{k+1}$ to be the string of $k$ zeros followed by $k$ pairs of $01$ and ending with $k+1$ zeros. 
    Define $Y=0^{k+1}(01)^k0^{k}$ to be $k+1$ zeros followed by $k$ pairs of $01$ and ending with $k$ zeros. 
    These two strings have Hamming distance $2k = (n-1)/2$. 
    
    Differentiating between $X$ and $Y$ is equivalent to determining the number of $0$s at the beginning or end of them (as this is a promise problem). It is known that it requires $\Omega(k) = \Omega(n)$ traces to determine if the length of the 0-run at the beginning is even or odd with probability at least $2/3$ because the problem reduces to differentiating between two binomial distributions~\cite{BatuKannan04-RandomCase}. Therefore, with probability at least $1/3$, a reconstruction algorithm using fewer traces must output a string that is at least Hamming distance $k = (n-1)/4$ away from the actual string. 
\end{proof}

\section{Conclusion}

We studied the challenge of determining the relative trace complexity of approximate versus exact string reconstruction.
Outputting a string close to the original in edit distance with few traces is a central problem in DNA data storage that has gone largely unnoticed in lieu of exact reconstruction. We present algorithms for classes of strings, where these classes lend themselves to techniques in every theoretician's toolbox (e.g., concentration bounds, estimates from averages), while introducing new alignment techniques that may be useful for other algorithms. Additionally, these classes of strings are hard to reconstruct exactly (they contain the set of $n$-bit strings with Hamming weight $n-1$, which suffices to derive an $\Omega(n)$ lower bound on the trace complexity).

 We left open the intriguing question of whether $\varepsilon n$-approximate reconstruction is actually easier than exact reconstruction for all strings. On the other hand, we showed that it is easier for at least some strings. Our algorithms output a string within edit distance $\varepsilon n$ from the original string using $O(\log n / \varepsilon^2)$ traces for large classes of strings. In some cases, we showed how to approximately reconstruct with a single trace. We also presented lower bounds that interpolate between the hardness of approximate and exact trace reconstruction.

Algorithms with small sample complexity for the approximate trace reconstruction problem could also provide insight into exact solutions. 
If we know that the unknown string belongs to a specified Hamming ball of radius $k$, then one can recover the string exactly with $n^{O(k)}$ traces by estimating the histogram of length $k$ subsequences~\cite{krasikov1997reconstruction, KMMP19}.
It is an open question whether an analogous claim can be proven for edit distance~\cite{grigorescu2020limitations}. Do $n^{O(k)}$ traces suffice if we know an edit ball of radius $k$ that contains the string? If this is true, then an algorithm satisfying our notion of edit distance approximation would imply an exact reconstruction result. 

Approximate trace reconstruction is also a specialization of {\em list decoding} for the deletion channel, where the goal is to output a small set of strings that contains the correct one with high probability. We are not aware of any work on list decoding in the context of trace reconstruction, even though it seems like a natural problem to study. Using an approximate reconstruction algorithm, we could output the whole edit ball around the approximate string. For more on list decoding with insertions and deletions, see the work by Guruswami, Haeupler, and Shahrasbi and references therein~\cite{guruswami2020optimally}.

\section{Acknowledgments}

We thank Jo\~ao Ribeiro and Josh Brakensiek for discussions on coded trace reconstruction, as well as the anonymous reviewers for helpful feedback on an earlier version of the paper.

{
\footnotesize
\linespread{.95}
\bibliographystyle{alpha}
\bibliography{refs}
}
\newpage
\appendix

\section{Appendix}\apndlab{A}

The following are omitted proofs from our warm-up approximate reconstruction algorithms.

\subsection{Analysis of first warm-up algorithm}

\begin{proof}[Proof of {\propref{longruns}}] 
    It is straight-forward to check that if $X$ contains $k$ runs, then with probability at least $1-\frac{1}{n^2}$
    all $T = \frac{2}{p \varepsilon^2} \log(n)$ traces contain $k$ runs. 
    Next, we estimate the lengths of runs in~$X$. For traces $\widetilde{X}_1,\ldots, \widetilde{X}_T$, 
    label the runs in $\widetilde{X}_j$ as $\widetilde{\mathbf{r}}^j_1,\widetilde{\mathbf{r}}^j_2,\ldots,\widetilde{\mathbf{r}}^j_k$, and recall that $|{\mathbf{r}}_i|$ denotes the length of the $i$th run, ${\mathbf{r}}_i$, in $X$.
    For $\widetilde{\mu}_i = \sum_{j=1}^T \widetilde{\mathbf{r}}^j_i/T$, the scaled average
    $\frac{\widetilde{\mu}_i}{p}$ estimates $|{\mathbf{r}}_i|$ for $i \in [k]$. Applying a Chernoff bound and then a union bound, 
    $\P(\exists i : |\widetilde{\mu}_i/p-|{\mathbf{r}}_i|| \ge \epsilon |{\mathbf{r}}_i|)  \leq 2n^{-3}$.
    Let $\widehat{X}=\widehat{X}_1\cdots \widehat{X}_k$, where substring $\widehat{X}_i$ is a run with length $\frac{\widetilde{\mu}_i}{p}$ and bit value matching run $i$ of the traces. 
    We have seen that with probability at least $1-\frac{1}{n}$, for every $i\in [k]$ the edit distance between $\widehat{X}_i$ and ${\mathbf{r}}_i$ is at most $\varepsilon |{\mathbf{r}}_i|$. On this event,     $\widehat{X}$ has edit distance at most $\epsilon n$ from $X$, by \lemref{substringrecon}.
\end{proof}

We can also achieve slightly stronger guarantees. If the number of traces in \propref{longruns} is linear, then the algorithm actually reconstructs {\em exactly} with high probability.  Also, the output~$\widehat{X}$ from the algorithm for \propref{longruns} will approximately reconstruct strings that do not quite satisfy the current assumptions, as described in the premises of the following corollary.





\begin{restatable}[Robustness]{corollary}{longrunsrobust}
\corlab{longrunsrobust}
    Let $X$ be an $n$-bit string such that all runs have length at least $\log(n^5)$ except for at most $s$ runs. We can $\epsilon n$-approximately reconstruct $X$ with $ O(\log(n) /\varepsilon^2 \cdot (\frac{1}{p})^{s})$ traces.
\end{restatable} 
\begin{proof} 
    Taking $C = 8/p$, with probability $1-\frac{1}{n^3}$ every long run (those with length at least $\log(n^4)$) will not be entirely deleted, and
    with probability at least $p^{s}$ none of the $s$ short runs are entirely deleted.
    By a Chernoff bound, with probability at least $1-n^{-3}$ the number of traces where no short run is entirely deleted is at least $\frac{3}{\epsilon^2p}\log(n)$.
    We identify the traces with the maximum number of runs and then use the
    algorithm for \propref{longruns} using these traces.
\end{proof}

\subsection{Analysis of second warm-up algorithm}

\begin{proof}[Proof of {\propref{oneruns}}]
    Suppose that all of the $1$-runs of $X$ have length at least $\frac{6}{p\epsilon^2}\log(n)$. 
    Take a single trace $\widetilde{X}$. 
    By a Chernoff bound, with probability at least $1-n^{-2}$, every $0$-run from $X$ with length at least 
    $\frac{6}{p\epsilon}\log(n)$ will have length at least $L := \frac{\log(n)}{10 \varepsilon}$ in $\widetilde{X}$. 
    Find every $0$-run in $\widetilde{X}$ with length at least $L$ and index them as $\widetilde{\mathbf{r}}_1$,\ldots,$\widetilde{\mathbf{r}}_k$. 
    For $i \in [k-1]$, let $\widetilde{s}_i$ be the bits between the last bit of $\widetilde{\mathbf{r}}_{i}$ and the first bit of $\widetilde{\mathbf{r}}_{i+1}$ and let $\widetilde{s}_0$ be the bits before $\widetilde{\mathbf{r}}_{1}$ and $\widetilde{\mathbf{s}}_{k+1}$ the bits after $\widetilde{\mathbf{r}}_{k}$.
    Let ${\mathbf{s}}_i$ be the contiguous substring of $X$ from which $\widetilde{\mathbf{s}}_i$ came and ${\mathbf{r}}_i$ the contiguous substring of $X$ from which $\widetilde{\mathbf{r}}_i$ came. 
    For all $i$, we will approximate ${\mathbf{s}}_i$ with $\widehat{1}_i,$ a $1$-run of length $|\widetilde{s}_i|/p$, and ${\mathbf{r}}_i$ with $\widehat{0}_i$, a $0$-run of length $|\widetilde{\mathbf{r}}_i|/p$.
    
    Since ${\mathbf{s}}_i$ contains alternating $1$-runs with length at least $\frac{6}{p\epsilon^2}\log(n)$ and $0$-runs with length at most $\frac{6}{p\epsilon}\log(n)$,
    ${\mathbf{s}}_i$ has at least a $1-\epsilon$ density of $1$s. 
    By a Chernoff bound, $\P\left(\left|\frac{|\widetilde{\mathbf{s}}_i|}{p} - |{\mathbf{s}}_i|\right| \ge \epsilon |{\mathbf{s}}_i|\right) \leq n^{-2}$. 
    Therefore $\widehat{1}_i$ 
    and ${\mathbf{s}}_i$ have edit distance at most $2\epsilon |{\mathbf{s}}_i|$. 
    If $|{\mathbf{r}}_i| \ge \frac{6}{p\epsilon^2}\log(n)$, then, as before,
    by a Chernoff bound $\P\left(\left||{\mathbf{r}}_i| - \frac{|\widetilde{\mathbf{r}}_i|}{p}\right| \ge \varepsilon|{\mathbf{r}}_i|\right) \le n^{-2}$, and so $\widehat{0}_i$ has edit distance at most $2\epsilon |{\mathbf{r}}_i|$ from ${\mathbf{r}}_i$. 
   If $|{\mathbf{r}}_i| \le \frac{6}{p\epsilon^2}\log(n)$ then the approximation of $|\widetilde{\mathbf{r}}_i|/p$ $0$s  has edit distance at most $ \frac{6}{ p \varepsilon} \log(n)$ from ${\mathbf{r}}_i$ 
    with probability at least $1-n^{-2}$.

    Let $\widehat{X}=\widehat{1}_0\widehat{0}_1 \widehat{1}_1\cdots \widehat{1}_k\widehat{0}_k \widehat{1}_{k+1}$ and
    observe that the number of $0$-runs is at most  $\frac{p\epsilon^2 n}{6\log(n)}$, since there at most this many $1$-runs which separate $0$-runs. 
    Then applying \lemref{substringrecon}, we have with probability at least $1-1/n$ that
    \begin{align*}
    \de(X,\widehat{X}) &\leq \sum_{i=1}^{k} (\de(\widehat{0}_i, {\mathbf{r}}_i) + \de(\widehat{1}_i, {\mathbf{s}}_i)) + \de(\widehat{0}_0, {\mathbf{s}}_0)+\de(\widehat{1}_{k+1}, {\mathbf{s}}_{k+1}) \\
    &\leq
    \sum_{i=1}^k \left (2 \varepsilon |{\mathbf{r}}_i|+\frac{6}{p \varepsilon} \log(n) +2 \varepsilon |{\mathbf{s}}_i| \right )+2 \varepsilon |{\mathbf{s}}_0|+2 \varepsilon |{\mathbf{s}}_{k+1}|\\ 
     &\leq
    2 \varepsilon n+\frac{6}{p \varepsilon} \log(n)\cdot \frac{n}{6\log(n)/(p\epsilon^2)}   \le 3 \varepsilon n.
    \end{align*}
    The theorem follows by taking
    $\epsilon = \frac{\epsilon^*}{3}$. 
\end{proof}

\section{Chernoff-Hoeffding Bound}

In many proofs, we use the following concentration bound:

\begin{lemma}[Chernoff-Hoeffding bound]
Let $X_1,\ldots, X_n \in \{0,1\}$ be independent. Let $b_1,\ldots,b_n \geq 0$ with $b = \max\{b_i\}$. 
Then for $0<\delta <1$, $X = \sum_{i=1}^n b_i X_i$, and $\mu = \mathbb{E}[X]$ the following holds:
$$
\P\left( \left|X - \mu\right| \geq \delta \mu\right) \leq 2\exp(- \mu \delta^2/(3b)).
$$
\end{lemma}

\end{document}